\documentclass[11pt]{article}
\usepackage[utf8]{inputenc}
\usepackage[T1]{fontenc}
\usepackage[english]{babel}

\usepackage{verbatim}
\usepackage{enumitem}
\usepackage{enumitem,kantlipsum}

\usepackage{amsmath}
\usepackage{amsthm}
\usepackage{amssymb}
\usepackage{latexsym}
\usepackage{staves}
\usepackage{marvosym}
\usepackage{eurosym}
\usepackage{phaistos}
\usepackage{verbatim}

\usepackage{color}
\usepackage{graphicx}

\usepackage{hyperref}
\hypersetup{
    colorlinks,
    citecolor=blue,
    filecolor=blue,
    linkcolor=blue,
    urlcolor=blue
}

\newtheorem{Theorem}{Theorem}[part]
\newtheorem{Definition}{Definition}[part]

\newtheorem{Remark}{Remark}[part]

\newtheorem{Ansatz}{Ansatz}[part]

\newtheorem{thm}{Theorem}[section]
\newtheorem{rem}[thm]{Remark}
\newtheorem{ans}[thm]{Ansatz}

\makeatletter \@addtoreset{equation}{section}

\@addtoreset{Definition}{section}

\@addtoreset{Theorem}{section}

\@addtoreset{Proposition}{section}

\@addtoreset{Property}{section}

\@addtoreset{Assumption}{section}

\@addtoreset{Corollary}{section}

\@addtoreset{Lemma}{section}

\@addtoreset{Remark}{section}

\@addtoreset{Example}{section}

\@addtoreset{Ansatz}{section}

\def\no{\noindent}

\def\05{\frac{1}{2}}
\def\-1{^{-1}}
\def\1{{1\hspace{-1mm}{\rm I}}}
\def\={\;=\;}
\def\.{\;.}

\title{  }
\author{ }

\def\be{\begin{eqnarray}}
\def\ee{\end{eqnarray}}

\def\b*{\begin{eqnarray*}}
\def\e*{\end{eqnarray*}}

\def\be{\begin{equation}}
\def\ee{\end{equation}}

\def\Esp#1{\mathbb{E}\left[#1\right]}

\def \E{\mathbb{E}}

\def\P{\mathbb{P}}
\def\Q{\mathbb{Q}}

\def\-1{^{-1}}

\def\0.5{\frac{1}{2}}

\def\no{\noindent}
\def\={\;=\;}
\def\.{\;.}

\def\1{{\bf 1}}

\title{No--Arbitrage Commodity Option Pricing \\ with Market Manipulation}

\author{Ren\'e A\"id\footnote{Université Paris-Dauphine, PSL Research University, LEDa. Email: rene.aid@dauphine.psl.eu.}  \quad Giorgia Callegaro\footnote{Universit\`a degli Studi di Padova, Dipartimento di Matematica, Via Trieste 63, I-35121 Padova, Italy. Email: gcallega@math.unipd.it.} \quad Luciano Campi\footnote{London School of Economics, Department of Statistics, Columbia House, 10 Houghton Street, WC2A 2AE, London, United Kingdom. Email: l.campi@lse.co.uk.}}

\date{\today}

\begin{document}

\maketitle

\begin{abstract}
We design three continuous--time models in finite horizon of a commodity price, whose dynamics can be affected by the actions of a representative risk--neutral producer and a representative risk--neutral trader. Depending on the model, the producer can control the drift and/or the volatility of the price whereas the trader can at most affect the volatility. The producer can affect the volatility in two ways: either by randomizing her production rate or, as the trader, using other means such as spreading false information. Moreover, the producer contracts at time zero a fixed position in a European convex derivative with the trader. The trader can be price-taker, as in the first two models, or she can also affect the volatility of the commodity price, as in the third model.

We solve all three models semi--explicitly and give closed--form expressions of the derivative price over a small time horizon, preventing arbitrage opportunities to arise. We find that when the trader is price-taker, the producer can always compensate the loss in expected production profit generated by an increase of volatility by a gain in the derivative position by driving the price at maturity to a suitable level. 
Finally, in case the trader is active, the model takes the form of a nonzero-sum linear-quadratic stochastic differential game and we find that when the production rate is already at its optimal stationary level, there is an amount of derivative position that makes both players better off when entering the game. \medskip

\no {\bf Keywords:} price manipulation, fair game option pricing, martingale optimality principle, linear-quadratic stochastic differential games. 
\end{abstract}


\clearpage
\section{Introduction} 
\label{sec:intro}

\begin{flushright}
{\em The methods and techniques of manipulation are limited only by the ingenuity of man}, \\ in Cargill vs Hardin, US Court of Appeal, 8$^{\rm th}$ circuit, Dec 7, 1971.
\end{flushright}

Price manipulation in financial markets is not the rare event as one may think. In their paper, Aggarwal and Wu \cite{Aggarwal06} provide data on more than 140 cases of market stock price manipulation in the sole period of ten years from 1990 to 2001 released by the Security Exchange Commission. As the authors quote, those cases only correspond to those who were caught.  On commodity markets, illegal practices of market manipulations are abundantly documented and can compete with stock markets (see Pirrong \cite{Pirrong17} for a survey of those practices). More recently, the LIBOR itself was the object of a coordinated manipulation by a cartel of banks. The LIBOR, created more than fifty years ago and controlled by the British Bankers Association, serves as a benchmark for loans and as an index in hundreds of trillions of nominal in derivatives. It is enough to read Duffie and Stein \cite{Duffie15} to measure the extent of social welfare loss induced by the manipulators actions. In their report for the Federal Reserve Bank of New York on the LIBOR scandal, Hou and Skeie \cite{Hou14} explain that if the first motivation for this manipulation in the aftermath of the 2008 financial crisis was to maintain a signal of credit worthiness, the second motivation was  the {\em express intent of benefiting the bank's derivatives positions}. 

Indeed, if the first generation of market price manipulation concentrated on using some market power to increase the market price and then resell the good at a higher price (unravelling strategy), it seems that the second generation of market manipulation will use the leverage effect provided by derivatives. Worrying enough to support this prognosis, the recent paper of Griffin and Shams \cite{Griffin17} asserts the possibility of an on--going VIX manipulation using large position in the out--of--the money options used to compute the VIX. If true, it would mean that some traders are already engaged in what was thirty years ago a theoretical problem in derivative pricing when academics would relax the hypothesis of no--market impact in the Black \& Scholes pricing framework (see  Jarrow \cite{Jarrow94} for a seminal work on this subject).\medskip

In this paper, we take market price manipulation models one step further in considering the possibility of the joint control of the average (the drift) and of the volatility of a commodity price by the actions of a producer and a trader who exchange a derivative. To analyse the behaviours of both players and the distortion of the prices of the commodity and of the derivative, we design three continuous--time models of increasing complexity. In each model, the commodity price is impacted by the actions of a representative risk--neutral producer and a representative risk--neutral trader. Both agents want to maximise their respective expected profits. The representative producer has market power and can increase or decrease the price by reducing or increasing her production rate. Actions on the volatility can be performed either by randomizing the production rate or by spreading false information. Production randomization is just making a strategic use of outages and the question answered in this paper is when this device has an interest for the producer. Further, regarding the use of information on the volatility, we suppose that  the trader or the producer has identified some channels allowing her to act on  the nominal volatility of the underlying by an appropriate {\em rate} of information. For both agents, manipulation of the commodity price comes at some costs, which are included in their profit functions.  

We consider first the case of production--based manipulation: the producer acts alone and can impact both the average price and its volatility by changing her average production rate and the volatility of her production rate. Second, we consider the case of production and information--based manipulation by a producer: the producer acts alone, she can affect the average price by changing her production rate and the volatility as mentioned above, e.g. by spreading appropriate information. Finally, we consider the case of a competition between a producer who can exert market power on the drift of the price and a trader who has an impact on the volatility of the price. In each case, we suppose that the producer contracts at time zero a constant (long or short) position of a European convex derivative  and delivers (or receives) its payoff at maturity. The trader has the opposite position in the derivative. 
Since they have an asymmetric impact on the dynamics of the commodity price, we are able to assess which instrument is more efficient, manipulation of the drift or manipulation of the volatility.  We aim at studying to which extent the producer and the trader can profit from their market power and how the prices of the commodity and its derivatives can be distorted by their actions.

In the classification of market manipulation provided by Allen and Gorton \cite{Allen92a}, the first model corresponds to  an {\em action--based} manipulation (using physical means such as production); the second model is a mixture of action--based and {\em information--based} (spreading false rumours on commodity scarcity or accounting and earnings' manipulation); the last one is a mixture of action--based, information--based and {\em trading--based} (buying to increase the price and then selling back).\medskip

To our knowledge, this is the first paper to study the joint manipulation of a commodity price and a derivative. We give now some reasons why our analysis might get even more relevant in a near future.

First, the commodity business has gone through a concentration trend in the last decades with the emergence of major players like Glencore/Xstrata, Rio Tinto or BHP Billinton. A small amount of international firms concentrate in their hands a significant volume of minerals production. For instance, Glencore concentrate 60\% of the zinc, 50\% of the copper, 45 \% of the lead and 38\% of the aluminium. At the same time, they have to take significant position in the financial markets to hedge their big physical positions. For instance, Glencore \cite[note 28, p. 201]{Glencore18} shows a position of \$3.2 billion of commodity related contracts including futures, options, swaps and physical forwards compared to an adjusted EBITDA of \$15.8 billion or a total asset value of \$128 billion. Rio Tinto \cite[notes p. 193]{RioTinto18} presents an exposition in nominal value of derivatives in aluminum of \$1.786 billion for an EBITDA for aluminum of \$3.1 billion and operating asset value of \$16.5 billion. Players of this size cannot ignore that the impact they may have on the price of a commodity will affect the value of their portfolio derivatives too.  

Second, large commodity firms are not the only big players in financial markets to hold significant positions in commodity derivatives. With the financialization of commodity markets, large hedge funds, banks and institutional players have increased their position in commodity derivatives (see Cheng and Xiong \cite{Cheng14} for an overview). Thus, when trying to move the price at her own advantage, a producer may find some opposition from financial actors harmed by her action. This problem is already documented in the case of large position of derivatives exchanged between financial institution (see the case of Merrill Lynch selling \$500 million of knock--in put options to Leiter's International in Gallmeyer and Seppi \cite{Gallmeyer00}). 

Third, the activities of trading in commodity firms are in general isolated in a subsidiary because they fall within the scope of financial regulation.  Thus, the trading activity might end up in conflict with the production activity regarding the use of market power. \medskip

Our main results can be summarized as follows. For each model we provide closed--form solutions in terms of a coupled Riccati systems of ordinary differential equations. In each model, the price of the derivative is a fair market price in the sense that it is consistent with no--arbitrage condition. First, we find that in all models, the optimal production rate of the producer follows the same pattern: during a transitory phase, it reaches the production rate that maximises the profit rate, then it stays there and at maturity, the production rate increases (resp. decreases) in case of short position (resp. long position). In the case of production--based manipulation, it is optimal for the producer to increase the volatility of the production rate to induce an increase in the volatility of the derivative if and only if he has a short position in the derivative exceeding a given threshold. Without derivative position, the value function of the producer is a non--increasing function of the volatility. Thus the producer prefers to reduce the volatility. But, when she holds a sufficiently large short position in the derivative, the increase in volatility that pushes the price of the derivative up can compensate the induced indirect cost of volatility. Since her impact on the average price is significant, the producer can compensate the loss in expected profit from production due to an increase of volatility by an increased profit from the derivative position. When the producer action on the volatility is information--based, the previous observations still hold, except that now her value function is increasing in the volatility, providing strong incentive to raise the volatility even without derivative position. Thus, it results that if the producer can impact the price to increase her profit in her derivative position, she does it.

What happens if the producer manipulation can be challenged by a trader taking an opposite position in the derivative? We find that the actions of the trader on the volatility only reduces the potential profit made by the producer on the derivative. Further, despite the asymmetry of powers of the two players, we find that when the production rate is already at its optimal stationary level, there is an amount of derivative position that makes both players better off entering the game. \medskip

There is a considerable financial economics literature about market manipulation, which follows in particular a game theoretic approach. A short review of the portion of such a literature related to stock markets has to start with the seminal work of Kyle \cite{Kyle85} and the work of Allen and Gale \cite{Allen92} who provide a simple information condition under which an uninformed trader can make a profitable unravelling strategy (buying the stock, making the price rise and sell the stocks at an average higher price). Chatterjea and Jarrow \cite{Chatterjea98} provide a game theoretic model of US Treasury Securities manipulation. Cooper and Donaldson \cite{Cooper98} design a dynamic game theoretic model of corner strategy. Regarding commodity price manipulation strategy, thorough analysis are available in the works of Pirrong \cite{Pirrong93, Pirrong95, Pirrong17}.

It is also worth mentioning that our modelling and contribution are different from those in the rich literature on market impact,
for which we refer the reader to, e.g., the recent book by Gu\'eant \cite{Gue16} and the references therein. Indeed, while in market impact models the drift of the market price is affected by the traders as a consequence of an optimal execution of a market order, in our setting both drift and volatility are affected and the impact comes directly from market manipulation. Moreover, the modelled financial phenomena are different and so are the problems solved (optimal execution vs profit maximization).

The closest work  to ours is the paper by Nyström and Parviainen \cite{Nystrom17}. The authors provide a zero--sum game between two players who can control the drifts and the volatilities of a multidimensional stock market and show under mild conditions that the game has a value and that it is given by the unique viscosity solution of degenerate parabolic PDE. Considering a more specific model of actors and impact functions, we are able to provide more insights in the gains of the producer and the trader and the distortion of prices.\medskip

 The paper is organised in the following way. Sections~\ref{sec:model1}, \ref{sec:model2} and \ref{sec:model3} present respectively the model of manipulation through production, manipulation through information and competition of manipulation. Section~\ref{sec:numeric} provides numerical illustration of the three models.

\section{Production--based manipulation} 
\label{sec:model1}

This section contains all our results on the first model of a producer of a commodity, who can manipulate the price of the commodity through production. 

More in detail, we consider a producer whose objective is to maximise her profit from production and from investment in a financial market over the time period $[0,T]$ for some $T>0$. The producer can increase her production rate $q_t$ with the instantaneous control rate $u_t$ at the expense of a cost $\frac{\kappa}{2} u^2$ with $\kappa >0$. We suppose that the production is entirely sold at a market price $\tilde S$, which can be affected by the producer: the more the production rate the less the market price. This effect leads to observed market price $\tilde S_t := s_0 - a \, q_t$ where $a > 0$ is some fixed parameter and $s_0 > 0$ is the market price before action of the producer (which in this case is constant). We will relax the hypothesis of a constant market price before impact in the second model (see Section \ref{sec:model2}). The former relation can also be seen as an inverse demand function of the good, where $a$ is its elasticity. Thus, the instantaneous profit rate is $P_t := q_t \, \tilde S_t$.

We suppose that the production rate $q_t$ is affected by a random factor that gathers all the randomness that usually affects production processes (outages, strikes and so on). We suppose that, without intervention of the producer, uncertainty is normally distributed with standard deviation $\sigma > 0$.  Further, we suppose that the producer has an effect on the uncertainty of his production rate. These hypotheses lead to the following dynamics for the production rate:
\be
\label{eq:qt}
dq_t = u_t \, dt +  \sigma \, \sqrt{1+z_t} \, dW_t, \quad q_0 \in \mathbb R,
\ee
where $W$ is a standard Brownian motion, defined on some probability space $(\Omega, \mathcal F,\mathbb P)$, and $z_t$ is the effect (in percentage) on the variance of the production rate. The information available to the producer is modelled by the natural filtration, ${(\mathcal F_t)}_{t \in [0,T]} = {(\mathcal F_t^W)}_{t \in [0,T]}$, generated by the Brownian motion $W$ and completed with all $\mathbb P$-null sets. Hence, anywhere in this section adaptedness will always be referred to this filtration.

We suppose that controlling $z_t$ requires some financial cost $\frac{g}{2} z^2$ with $g >0$. The producer can choose either to decrease or to increase the volatility of the production rate $q_t$. Although the costs incurred to increase the volatility are less easy to grasp than the cost involved to decrease it, they can be interpreted as the costs of the actions needed to hide them.

At this stage of the model description, we notice that an increase of the volatility of $q_t$ has a negative impact on the expected instantaneous profit $\mathbb E[P_t] = s_0 \mathbb E[q_t] - a\mathbb E[q_t ^2]$. In other terms, the producer is Gamma negative. Thus, he has no incentive to increase the volatility of his production facilities.

Most large commodity producers make an important use of financial market for hedging purposes. Hence, we suppose that the producer intervenes in the financial market for his production good by selling derivatives at the initial time. Since the producer is Gamma negative, a natural hedge would be to sell a Gamma positive derivative such as a call option. Here, we suppose that the producer has a net derivative position $\lambda$ which can be positive (short, sale) or negative (long, purchase) with maturity $T$ and payoff $h_T := \tilde S_T^2$. Such a quadratic payoff can be seen as a position over a portfolio of call options with the same maturity $T$ and different strike prices. We denote by $h_0$ the price at time $0$ of that option, its precise definition will be given when specifying the set of admissible controls. Indeed, $h_0$ is not given from the outset, as it depends on the underlying which is in turn controlled.  

The aim of the producer is to maximise the following objective functional 
\be
\label{eq:v}
J^\lambda (u,z,h_0) := \Esp{ \int_0^{T} \left( P_t   - \frac{\kappa}{2}u^2 _t - \frac{g}{2}z^2 _t \right) dt + \lambda \big( h_0 - h_T\big)  }.
\ee
Now, within the producer firm, there are two distinct departments, a production department and an investment department. The former takes any production--related decisions, i.e. it controls $u$ and $z$, while the latter is responsible for selling/buying the derivatives at a fair price and pursuing the corresponding hedging strategy. It is natural to assume that the investment department is using the nowadays classical no-arbitrage machinery to propose a derivative's price. The two departments 
are aware of the fact that their decision affects each other. In particular, the fact that the investment department uses the no-arbitrage approach to price derivatives implies the following natural constraints for the production side: the chosen production plan should not lead to arbitrage opportunities.
We are going to incorporate this idea into the definition of admissible policies. After that, we will give the precise formulation of the producer optimization problem.   

\begin{Definition}\label{adm-case1} We say that any pair $(u,z)$ is \emph{admissible} if the following properties are satisfied:\begin{enumerate}
\item[(i)] $(u_t ,z_t )_{t \in [0,T]}$ are progressively measurable processes with values in $\mathbb R \times (-1,\infty)$ such that
\[ \mathbb E\left[\int_0 ^T (u_t ^2 + z^2 _t)dt\right] < \infty, \quad \mathbb E\left[\int_0 ^T q_t ^2 (1 + z_t) dt\right] < \infty; \]
\item[(ii)] there exists a unique equivalent martingale measure $\mathbb Q^{u,z}$ for the price process $\tilde S$, equivalently for the production process $(q_t)_{t\in [0,T]}$, with $h_T \in L^1 (\mathbb P) \cap L^1 (\mathbb Q^{u,z})$; 
\item[(iii)] there exists a real-valued progressively measurable process $(\Delta^{u,z}_t)_{t\in [0,T]}$ satisfying
\[ \mathbb E\left[ \int_0^T \left( |\Delta^{u,z}_t u_t| + |\Delta^{u,z}_t |^2 (1+z_t) \right) dt \right] < \infty ,\]
and such that the following holds $\mathbb Q^{u,z}$-a.s.
\[ h^{u,z}_t := \E^{\Q^{u,z}} [ h_T| \mathcal F_t] = \E^{\Q^{u,z}} [ h_T]   + \int_0 ^t \Delta^{u,z}_s dq_s,\] 
for all $t \in [0,T]$. \end{enumerate}
The set of all admissible pairs will be denoted by $\mathcal A$. 
\end{Definition}

Hence $h^{u,z}_0 = \E^{\Q^{u,z}} [ h_T] $ can be viewed as the price of the option $h_T$ under the production control $(u,z)$. Notice that such a price is clearly affected by the controls via the risk neutral measure in (ii). It can be interpreted as ``commitment price'': after selling the option, the producer could deviate from the implementation of the hedging strategy that leads to the measure $\Q^{u,z}$.  Here we make the assumption that the producer implements the production controls leading to precisely that measure and thus, that price.

Notice that from $q$'s dynamics we have that the measure $\Q^{u,z}$, whose existence is postulated in (ii) above, is necessarily given by the following Radon-Nikodym derivative
\[ \frac{d\Q^{u,z}}{d\P} = \exp\left\{  -\int_0 ^T \delta_t dW_t - \frac{1}{2}\int_0 ^T \delta_t ^2 dt \right\}, \quad \delta_t = \frac{u_t}{\sigma \sqrt{1+ z_t}}. \]
Before giving the final formulation of the producer optimization problem, we can exploit the admissibility properties above to rewrite the objective functional \eqref{eq:v} as follows
\[ J^\lambda (u,z,h_0) = \Esp{ \int_0^{T} \left( P_t   - \frac{\kappa}{2}u^2 _t - \frac{g}{2}z^2 _t - \lambda \Delta_t^{u,z} u_t \right) dt} = : \tilde J^\lambda (u,z) . \]
Indeed, condition (iii) implies that
\[ h_T - h_0 = \int_0 ^T \Delta_t^{u,z}  dq_t = \int_0 ^T \Delta_t^{u,z}  (u_t dt + \sigma \sqrt{1+z_t}dW_t).\]
Condition (iii) implies that the $dW$-part above has zero expectation under $\mathbb P$. Moreover, we observe that the integrability assumptions in (i) and (iii) ensure that $\tilde  J^\lambda (u,z)$ is finite. 

Finally, after all these preliminaries, we can formulate the producer's optimization problem
\begin{equation}\label{maxJ}
\sup_{(u,z) \in \mathcal A} \tilde J^\lambda (u,z).
\end{equation}

\subsection{Heuristics} 
\label{sec:solution}

In this part we develop the heuristics needed to obtain a candidate for the solution of problem \eqref{eq:v}. In the next sub-section, we will verify that the candidate is indeed the optimal solution according to the definition above. 

First, notice that, since the market is complete, there exists only one possible no-arbitrage price for the derivative $h_T$, which also gives the initial wealth needed to fund the hedging strategy. The derivative can be perfectly replicated by trading in a self-financing way in the underlying $\tilde S_t = s_0 -aq_t$ or, equivalently, in $q_t$. Therefore
\[ h_T = \mathbb E^\Q [h_T] + \int_0 ^T \Delta_t dq_t ,\]
where $\mathbb Q$ is the unique equivalent martingale measure for $\tilde S$ (or, equivalently, for $q$), and $\Delta$ is the delta hedging. More precisely, using Girsanov's theorem we get the dynamics of $q$ under $\Q$, which is
\[ dq_t = \sigma \sqrt{1+z_t} dW_t ^\Q, \quad dW^\Q_t = dW_t - \delta_t dt,\]
where $\delta_t = \frac{u_t}{\sigma \sqrt{1+z_t}}$ for $t\in [0,T]$. Hence, defining the price at time $t$ of the derivative as
\[ h_t = \mathbb E^\Q [ (s_0 -aq_T)^2 | q_t] := \varphi(t,q_t),\]
we obtain the usual PDE for the price
\begin{equation}\label{price-PDE1}
\varphi_t + \frac{1}{2} \sigma^2 (1+z(t,q)) \varphi_{qq} = 0, \quad \varphi(T,q)= (s_0 -aq)^2.
\end{equation}
Finally, we have the usual relationship $\Delta_t = \varphi_q (t,q_t)$.

\begin{Remark} \label{Markov_z}
{\rm Notice from Equation \eqref{price-PDE1} that $\varphi$ depends on $z$ in a functional way.
However we expect the optimal control to be Markovian, which justifies replacing $z_t$ (which could in principle depend on the whole path of the state variable $(q_t)$) with the function $z(t,q)$ of time and of the value $q$ of state variable at time $t$. The PDE above needs to be solved together with the HJB equation for the value function, since the coefficient of the second derivative $\varphi_{qq}$ depends on the control $z$. }
\end{Remark}

The perfect replicability of the derivative allows us to rewrite the objective function in a more suitable way as at the end of the previous sub-section. Indeed, observe first that
\[ h_0 - h_T = - \int_0 ^T \Delta_t dq_t = -\int_0 ^T \varphi_q (t,q_t) (u_t dt + \sigma \sqrt{1+z_t} dW_t), \]
where recall that $W$ is a Brownian motion under $\mathbb P$. Hence, given the hedging strategy $\varphi_q(t,q_t)$, the maximization problem on the production side can be expressed as follows
\begin{equation}
v^\lambda (0,q_0) := \sup_{u,z} \mathbb E \left [ \int_0 ^T \left(q_t (s_0 -aq_t) - \frac{g}{2}z_t ^2 - \frac{\kappa}{2}u_t ^2 - \lambda \varphi_q (t,q_t) u_t \right)dt \right].
\end{equation}
In other terms, the gain coming from selling the derivative has been absorbed by the running profit term. Therefore, we can get the HJB equation
\begin{eqnarray}\label{HJB-case1}
-v_t = \sup_{u,z} \left\{ q(s_0 -a q) - \frac{g}{2}z ^2 - \frac{\kappa}{2}u ^2 - \lambda \varphi_q (t,q) u + u v_q + \frac{\sigma^2}{2}(1+z) v_{qq}\right \},
\end{eqnarray}
with terminal condition $v(T,q)=0$. Notice that the PDE for the price \eqref{price-PDE1} and the HJB equation for the value function are clearly coupled as the optimal $z$ appears in the pricing PDE, while the derivative of the price, $\varphi_q$, appears in the HJB equation (compare to Remark \ref{Markov_z}).
The first order conditions give the two (candidate) optimal controls
\begin{equation}
\widehat u = \frac{1}{\kappa}\left(v_q-\lambda \varphi_q\right), \quad \widehat z =\frac{\sigma^2}{2g} v_{qq}.
\end{equation}
Notice that we have dropped the dependence upon $\lambda$ in the value function for sake of readability. In order to get the full solution, it is natural to make the following

\begin{Ansatz}\label{ans-case1}
Both solutions $\varphi$ and $v^\lambda$ are quadratic in $q$, i.e.
\[ \varphi(t,q)= A(t)q^2 + B(t)q+C(t), \quad v^{\lambda}(t,q)=D(t)q^2 + E(t)q +F(t),\]
where $A,B,C,D,E,F$ are deterministic functions of time, to be determined.
\end{Ansatz}

\paragraph{Solving for $\varphi$.} To ease the notation, we drop the dependence of time from $A,B$ and so on. First of all, applying the Ansatz \ref{ans-case1} to the candidate optimal controls gives
\begin{equation*} \label{opt-uz} \widehat u = \frac{1}{\kappa}\left( 2q(D-\lambda A) +E -\lambda B \right), \quad \widehat z = \frac{\sigma^2}{g}D.\end{equation*}
Next, we substitute the expression above for $\widehat u$ and $\widehat z$ in the pricing PDE \eqref{price-PDE1} and we obtain
\begin{equation*}
\sigma^2 \left(1+\frac{\sigma^2}{g}D\right) A + A' q^2 + B' q + C' =0, \quad A(T)q^2 + B(T) q +C(T) = s_0 ^2-2as_0 q +a^2q^2.
\end{equation*}
In particular, the terminal condition for $\varphi$ gives the corresponding terminal conditions for $A,B,C$ as
\begin{equation*} A(T) = a^2, \quad B(T)=-2as_0 , \quad C(T)=s_0 ^2.\end{equation*}
By identification of the terms in $q$, we get the following ODEs for $A,B$ and $C$
\[ A'=0, \quad B'=0, \quad C'=-\sigma^2a^2\left( 1+ \frac{\sigma^2}{g} D\right),\]
which can be easily solved using the terminal conditions above. Indeed, we obtain
\begin{equation}\label{sol-ABC} 
A(t)=a^2, \quad B(t) = -2as_0 , \quad C(t) = s_0 ^2 + \sigma^2 a^2 \int_t ^T \left( 1+ \frac{\sigma^2}{g}D(r)\right) dr,
\end{equation}
for all $t \in [0,T]$. Notice that the function $D(t)$ will be obtained when solving the HJB equation \eqref{HJB-case1}.
 
\paragraph{Solving for $v^\lambda$.}  Substituting the Ansatz \ref{ans-case1} for $v^\lambda$ (together with the optimal controls) in the HJB equation \eqref{HJB-case1} and identifying the terms in $q$, we obtain the following ODEs for $D,E$ and $F$
\begin{eqnarray*}
-D' &=& -a +\frac{2}{\kappa}(D-\lambda A)^2, \quad D(T)=0,\\
-E' &=& s_0 + \frac{2}{\kappa}(D-\lambda A) (E-\lambda B), \quad E(T)=0,\\
-F' &=& \frac{\sigma^4}{2g}D^2 + \sigma^2 D +   \frac{1}{2 \kappa}(E-\lambda B)^2 , \quad F(T)=0.
\end{eqnarray*}
Now, using (\ref{sol-ABC}) implies
\begin{eqnarray}
-D' &=& -a +\frac{2}{\kappa}(D-\lambda a^2)^2, \label{eqD} \\
-E' &=& s_0 + \frac{2}{\kappa}(D-\lambda a^2) (E + 2\lambda as_0 ),\\
-F' &=& \frac{\sigma^4}{2g}D^2 + \sigma^2 D +\frac{1}{2 \kappa}(E + 2 a \lambda s_0 )^2 ,
\end{eqnarray}
with null terminal conditions $D(T)=E(T)=F(T)=0$.

\begin{Remark}\label{rem:sign-D}
{\rm We observe that, while the equation for $D$ is a one-dimensional Riccati ODE, the second one is linear and the third one can be solved just by integration. The Riccati equation (\ref{eqD}) can be easily proved to have a unique solution over the whole time interval $[0,T]$. Indeed, this is a direct consequence of, e.g., Lemma 10.12 in \cite{Filipovic09}. Moreover, that lemma also implies that
\[ D(t) \le 0 \quad \text{for all }t \in [0,T] \quad \Leftrightarrow \quad D'(T) = a - \frac{2\lambda^2 a^4}{\kappa} \ge 0,\]
which will be important later for the interpretation of our results. The value $a - \frac{2\lambda^2 a^4}{\kappa}$ corresponds to the slope of $D$ close to $T$.}
\end{Remark}

Set $\theta = \sqrt{8a / \kappa}$. Solving the equations above gives the following expressions
\begin{eqnarray}
D(t) &=& - \frac{2(a-\frac{2\lambda^2 a^4}{\kappa})(e^{\theta (T-t)} -1)}{\theta (e^{\theta(T-t)} +1) +\frac{4\lambda a^2}{\kappa} (e^{\theta(T-t)} -1)}, \label{eq-D} \\
E(t) &=& s_0 \int_t ^{T} e^{\int_t ^{u} \frac{2}{\kappa}(D(r)-\lambda a^2)dr} \left[ 1 +\frac{4 a \lambda}{ \kappa}(D(u)-\lambda a^2 ) \right] du, \label{eq-E} \\
F(t) &=&  \int_t ^T \left(\frac{\sigma^4}{2g}D(u) ^2 + \sigma^2 D(u) + \frac{1}{2 \kappa}  {{(E(u) + 2a \lambda s_0 )}^2} du \right) , \nonumber \\
&& \label{eq-F}
\end{eqnarray}
for all $t \in [0,T]$.

\subsection{Verification}
We conclude the section with the verification that the candidate described above is indeed a solution to problem \eqref{maxJ}.
\begin{Theorem}\label{verif-case1}
Let $D,E$ be deterministic functions of time as in, respectively, \eqref{eq-D} and \eqref{eq-E}. Whenever $H:= 1- \frac{2\lambda a}{\kappa}(\lambda a^2 - \frac{g}{\sigma^2})>0$, we assume that the maturity $T$ is small enough, more precisely
\begin{equation}\label{small-T} T < T_{max} := \frac{2}{\theta} \coth^{-1} \left( \frac{2\sigma^2 a}{\theta g} H\right). \end{equation}
Then there exists an optimal policy $(\widehat u, \widehat z) \in \mathcal A$ for problem \eqref{maxJ}, where the production policies are
\begin{equation}
\label{opt-uz-2} \widehat u_t = \frac{1}{\kappa}\left( 2 \widehat q_t (D(t)-\lambda a^2 ) + 2as_0 \lambda +E(t) \right), \quad \widehat z_t = \frac{\sigma^2}{g}D(t), \quad t\in [0,T].\end{equation}
The no-arbitrage price process for the derivative $h_T$ is given by
\begin{equation}\label{h-price} 
\widehat h_t := h_t^{\widehat u, \widehat z}= (s_0 -a\widehat q_t)^2 + \sigma^2 a^2 \int_t ^T \left( 1+ \frac{\sigma^2}{g}D(u)\right) du, \quad t \in [0,T],
\end{equation}
and the hedging process is
\begin{equation}
\widehat \Delta _t := \Delta_t^{\widehat u, \widehat z} = 2 a (a \widehat q_t - s_0), \quad t\in [0,T],
\end{equation}
where the production rate is
\begin{eqnarray}
	\widehat q_t &=& e^{R(t)}\left\{ q_0 + \int_0^t e^{-R(s)} \frac{1}{\kappa}(-2\lambda a^2 + 2as_0 \lambda + E(s))ds \right. \nonumber \\
	&& \left.+ \int_0 ^t e^{-R(s)} \sigma \sqrt{1+ \frac{\sigma^2}{g}D(s)} dW_s \right\}, \label{q_hat_pb1}
\end{eqnarray}
with $R(t) := \int_0 ^t \frac{2}{\kappa} D(s)ds$.
\end{Theorem}

\begin{proof}
The proof is structured in two steps.\begin{enumerate}[wide, labelwidth=!, labelindent=0pt]
\item \emph{Admissibility}. Let us verify that the pair $(\widehat u,\widehat z)$ given in the statement above belongs to $\mathcal A$. We start from condition (i) in Definition \ref{adm-case1}. First, $\widehat u, \widehat z$ are trivially progressively measurable and real valued. We need to check $\widehat z_t > -1$, which is equivalent to $\frac{\sigma^2}{g}D(t) > -1$. We distinguish two cases: if $a \le \frac{2\lambda^2a^4}{\kappa}$ we have $D(t) \ge 0$ (this is a consequence of Remark \ref{rem:sign-D} or, alternatively, the explicit formula \eqref{eq-D}), hence $\widehat z_t >-1$ for all $t \in [0,T]$. On the other hand, if $a > \frac{2\lambda^2a^4}{\kappa}$, it follows from expression \eqref{eq-D} that $D(t)$ is nondecreasing with $D(T)=0$. Therefore, it suffices to check that $D(0)> -g/\sigma^2$, where
\[ D(0)= -\frac{2 (a - \frac{2\lambda^2a^4}{\kappa})}{\theta \coth (\frac{\theta T}{2}) + \frac{4\lambda a^2}{\kappa}}.\]
After some computation, we obtain that $D(0)> -g/\sigma^2$ if and only if
\[  \coth \left(\frac{\theta T}{2}\right) > \frac{2\sigma^2 a}{\theta g} H,\]
where $H$ is the constant defined in the statement. Now, if $H <0$ the inequality above is always satisfied as the LHS above is nonnegative. If $H>0$, the inequality above is guaranteed by the condition $T < T_{max}$.
We can conclude that even in this second case, provided $T< T_{max}$, we have $\widehat z_t >-1$ for all $t\in [0,T]$. Regarding the integrability properties, we verify now that
\begin{equation}\label{square-int-sol} \mathbb E\left[\int_0 ^T (\widehat u_t ^2 + \widehat z^2 _t)dt\right] < \infty, \quad \mathbb E\left[\int_0 ^T \widehat q_t ^2 (1 + \widehat z_t) dt\right] < \infty, \end{equation}
where $\widehat q = q^{\widehat u,\widehat z}$. Since $\widehat u$ is affine in $\widehat q$ with continuous time-dependent coefficients and $\widehat z$ is deterministic and continuous in $t$, checking the properties above boils down to show
\[ \mathbb E\left[\int_0 ^T \widehat q_t ^2 dt\right] < \infty.\]
First, we use Fubini's theorem to get $\mathbb E [\int_0 ^T \widehat q_t ^2 dt ] = \int_0 ^T \mathbb E[\widehat q_t ^2] dt$. Moreover, since $\widehat q_t$ is a Gaussian random variable for any fixed $t$ (see Remark \ref{explicit-q} below), the function $t \mapsto \mathbb E[\widehat q_t ^2]$ is continuous over $[0,T]$, so its integral is finite.  

\noindent Regarding condition (ii), we need to show that there exists a unique EMM $\widehat{\mathbb Q} = \mathbb Q^{\widehat u,\widehat z}$ for the production process $\widehat q$. Let us recall that
\[ \frac{d\widehat \Q}{d\P} = \exp\left\{  -\int_0 ^T \widehat \delta_t dW_t - \frac{1}{2}\int_0 ^T \widehat \delta_t ^2 dt \right\}, \quad \widehat \delta_t = \frac{\widehat u_t}{\sigma \sqrt{1+ \widehat z_t}}. \]
We use \cite[Theorem 2.1]{Rydberg97} to prove that under our assumptions the probability $\widehat \Q$ is well-defined (see also \cite{Ruf15} for more general results of the same type). According to that results, we need to check Assumption 2.2 in \cite{Rydberg97}, which in our case is satisfied as long as $\sigma^2 (1+\widehat z_t) > 0$ for all $t\in [0,T]$. By the same arguments used for condition (i), we get the result. A standard application of Girsanov theorem, together with the integrability properties in \eqref{square-int-sol}, yields immediately that $\widehat q$ is a martingale under $\widehat \Q$.

To end checking condition (ii), we have to show $h_T \in L^1(\mathbb P) \cap L^1(\widehat \Q)$. Now, $h_T = (s_0 -a\widehat q_T)^2$, hence quadratic in $\widehat q_T$. Since under both probability measures $\widehat q_T$ is a Gaussian random variable, we have $q_T ^2 \in  L^1(\mathbb P) \cap L^1(\widehat \Q)$, which gives the desired property.

\noindent We pass to condition (iii) in Definition \ref{adm-case1}. First, $\widehat \Delta$ is trivially a progressively measurable process with real values. For the integrability property, since both $\widehat u$ and $\widehat \Delta$ are linear in $\widehat q_t$, we are again reduced to the square integrability $\mathbb E[\int_0 ^T \widehat q_t ^2 dt]<\infty$, which has been proved just before. 

To conclude this part of the proof, it remains to check that, given $(\widehat u, \widehat z)$ as above, $\widehat h_t := h^{\widehat u,\widehat z}_t = \E^{\widehat \Q} [ h_T| \mathcal F_t] = \E^{\widehat \Q} [ h_T]   + \int_0 ^t \widehat \Delta_s dq_s$ a.s. under $\widehat \Q$, for all $t \in [0,T]$. This can be done by direct computation as follows: applying It\^o's formula to $\widehat h_t$ in \eqref{h-price} we get
\[ d\widehat h_t = -2a(s_0 -a\widehat q_t)d\widehat q_t \]
whence, in integral form, 
\[\widehat h_t = \widehat h_0 -2a \int_0 ^t  (s_0 -a\widehat q_s)d\widehat q_s = \widehat h_0 + \int_0 ^t \widehat \Delta_s d\widehat q_s.\] 
Moreover, one easily find
\[ \widehat{\mathbb E} [h_T] = \widehat{\mathbb E} [(s_0 -a\widehat q_T)^2] = s_0^2 - 2as_0 q_0 + a^2 \widehat{\mathbb E}[\widehat q_T ^2],\]
where $\widehat{\mathbb E}$ denotes the expectation with respect to the measure $\widehat \Q$.
Using It\^o's isometry, we also have 
\[ \widehat{\mathbb E}[\widehat q_T ^2]  =  q_0 ^2 +  \int_0 ^T \sigma^2 \left(1+\frac{\sigma^2}{g}D(t)\right) dt,\]
which leads to the remaining property in (iii).

\item \emph{Optimality}. To check the optimality condition, we are going to use the martingale optimality principle (see \cite{elkaroui-stflour}). Let us define the process
\begin{equation}\label{def-Y}
Y^{u,z}_t := \int_0 ^t  \left(q_r (s_0 -aq_r) - \frac{g}{2}z_r ^2 - \frac{\kappa}{2}u_r ^2 - \lambda \widehat \Delta_r u_r \right)dr + V(t,q_t),
\end{equation}
with
\[ V(t,q) = D(t)q^2 + E(t)q + F(t).\]
Thanks to the martingale optimality principle, proving that $Y^{u,z}$ is a supermartingale for all $(u,z) \in \mathcal A$ and a martingale for $(u,z)=(\widehat u,\widehat z)$, will give us the result. It\^o's formula yields
\begin{align} 
dY^{u,z}_t =& \left[q_t (s_0 -aq_t) - \frac{g}{2}z_t ^2 - \frac{\kappa}{2}u_t ^2 - \lambda \widehat \Delta_t u_t + V_t + V_q u_t + \frac{1}{2}V_{qq}\sigma^2 (1+ z_t) \right]dt \nonumber \\
& + V_q \sigma \sqrt{1+z_t} dW_t, \label{Y-dyn}
\end{align}
where $V_t,V_q,V_{qq}$ denote partial derivative of the function $V(t,q)$. We have omitted the dependence on $(t,q)$ for sake of simplicity. First, notice that due to the integrability properties in Definition \ref{adm-case1}(i)  the process $\int_0 ^t V_q \sigma \sqrt{1+z_r} dW_r$ is a true $\mathbb P$-martingale. Therefore, it remains to show that the $dt$-part in \eqref{Y-dyn} above is a nonincreasing process, that is it is lower or equal to zero almost everywhere. By construction (see heuristics), we know that the function $V(t,q)$ satisfies the HJB equation \eqref{HJB-case1}, with
\[ \varphi_q (t,q) = 2a^2 q -2a s_0 ,\]
hence $\varphi_q (t,q_t) = \widehat \Delta_t$ for all $t$. We can conclude that the drift in \eqref{Y-dyn} is lower or equal to zero a.e., yielding that $Y^{u,z}$ is a supermartingale for all $(u,z) \in \mathcal A$. To show that $Y^{\widehat u,\widehat z}$ is a martingale, one proceeds in the same way getting equalities instead of inequalities. In particular one gets that the drift in \eqref{Y-dyn} is equal to zero a.e., implying the martingale property.
\end{enumerate}
Finally, we can solve for $\widehat q_t$ in explicit form, via standard resolution methods, since it is the unique solution of the following linear SDE
	\[ d\widehat q_t = \frac{1}{\kappa}\left( 2\widehat q_t (D(t)-\lambda a^2 ) + 2as_0 \lambda +E(t) \right) dt + \sigma \sqrt{1+\frac{\sigma^2}{g} D(t)} dW_t , \quad \widehat q_0 = q_0.\] 
\end{proof}

\begin{Remark}\label{explicit-q}
{\rm 
Notice, from Equation \eqref{q_hat_pb1}, that $\widehat q_t$ is a Gaussian random variable with time-dependent mean and variance.
}
\end{Remark}

\section{Production and information based manipulation}
\label{sec:model2}

In this section we describe and solve explicitly a variant of the model presented before. We still have a producer of a commodity, who is maximizing her profit coming from both production and a short/long position in some derivative. The main differences are that the market price of the commodity is no longer a constant as it is driven by the Brownian motion $W$, that in turn does not affect the production rate anymore and, finally, the producer can directly control the volatility of the market price (by spreading false information on the state of his production, for instance). Thus, the dynamics of the state variables is now given by
\begin{equation}\label{eq:sq}
\left\{
\begin{array}{rcl}
dS_t & = & \mu \, dt +  \sigma \, \sqrt{1+z_t} \, dW_t, \\
dq_t & = & u_t \, dt.
\end{array}
\right.
\end{equation}
Moreover in this model the market price $\tilde S$ is given by $\tilde S_t = S_t -aq_t$, $t \in [0,T]$. The objective of the producer is the same as before, i.e.
\be
\label{eq:v2}
J^\lambda (u,z,h_0) := \Esp{ \int_0^{T} \left( P_t   - \frac{\kappa}{2}u^2 _t - \frac{g}{2}z^2 _t \right) dt + \lambda \big( h_0 - h_T\big)  }.
\ee
 Analogously to the previous model, we will be working with the following definition of admissible policies, which admits the same interpretation as before.
\begin{Definition} 
\label{adm-case2} 
We say that any pair $(u,z)$ is \emph{admissible} if the following properties are satisfied:
\begin{enumerate}
\item[(i)] $(u_t ,z_t )_{t \in [0,T]}$ are progressively measurable processes with values in $\mathbb R \times (-1,\infty)$ such that
\[ \mathbb E\left[\int_0 ^T (u_t ^2 + z^2 _t)dt\right] < \infty, \quad \mathbb E\left[\int_0 ^T \tilde S_t ^2 (1 + z_t) dt\right] < \infty; \]
\item[(ii)] there exists a unique equivalent martingale measure $\mathbb Q^{u,z}$ for the price process $\tilde S$, with $h_T \in L^1 (\mathbb P) \cap L^1 (\mathbb Q^{u,z})$; 
\item[(iii)] there exists a real-valued progressively measurable process $(\Delta_t^{u,z})_{t\in [0,T]}$ satisfying
\[ 
\mathbb E\left[ \int_0^T \left( |\Delta_t^{u,z} u_t| + |\Delta_t^{u,z} |^2 (1+z_t) \right) dt \right] < \infty,
\]
and such that the following holds $\mathbb Q^{u,z}$-a.s.
\[ h_t^{u,z} := \E^{\Q^{u,z}} [ h_T| \mathcal F_t] = \E^{\Q^{u,z}} [ h_T]   + \int_0 ^t \Delta_s^{u,z} d \widetilde S_s ,\]
for all $t \in [0,T]$. \end{enumerate}
The set of all admissible pairs will be denoted by $\mathcal A$. 
\end{Definition}

Analogously as in the previous model, from $\tilde S$'s dynamics we have that the measure $\Q^{u,z}$, whose existence is postulated in (ii) above, is necessarily given by the following Radon-Nikodym derivative
\[ \frac{d\Q^{u,z}}{d\P} = \exp\left\{  -\int_0 ^T \gamma_t dW_t - \frac{1}{2}\int_0 ^T \gamma_t ^2 dt \right\}, \quad \gamma_t = \frac{\mu - au_t}{\sigma \sqrt{1+ z_t}}. \]
Before giving the final formulation of the producer's optimization problem in this model as well, we can exploit the admissibility properties above to rewrite the objective functional \eqref{eq:v2} as follows
\[ J^\lambda (u,z,h_0) = \Esp{ \int_0^{T} \left( P_t   - \frac{\kappa}{2}u^2 _t - \frac{g}{2}z^2 _t - \lambda \Delta_t^{u,z} (\mu -au_t) \right) dt} = : \tilde J^\lambda (u,z) . \]
Finally, the producer's optimization problem, that we are going to solve in the next sub-section, is given by
\begin{equation}\label{maxJ-case2}
\sup_{(u,z) \in \mathcal A} \tilde J^\lambda (u,z).
\end{equation}

\subsection{Heuristics} 
In this sub-section we describe the heuristics that led us to propose some candidate solution. It follows the same lines as in the first model. Being the market complete, there exists a unique martingale measure $\mathbb Q$ under which $\widetilde S$ is a martingale. Indeed we have
\begin{eqnarray*}
d \widetilde S_t &=& (\mu - a u_t) dt + \sigma \sqrt{1+z_t} d W_t \\
& = & (\mu - a u_t - \gamma_t \sigma \sqrt{1+z_t}) dt + \sigma \sqrt{1+z_t} d W_t^{\mathbb Q},
\end{eqnarray*}
where $W^{\mathbb Q}$ is a $\mathbb Q$-Brownian motion and by choosing $\gamma_t = \frac{\mu - a u_t}{\sigma \sqrt{1+z_t}}$, $\widetilde S$ is a $\mathbb Q$-martingale.
The price at time $t=0$ of the European claim $h_T$ equals $\mathbb E^{\mathbb Q} [  \widetilde S_T^2 ]$, so that we expect that $ h_T = h_0 + \int_0^T \Delta_t d \widetilde S_t $, where $\Delta$ is the corresponding delta hedging strategy. Hence, the difference $h_0 - h_T $ also reads
\begin{equation}\label{DeltaH}
h_0 - h_T  = - \int_0^T \Delta_t d \widetilde S_t.
\end{equation}
The last quantity to be introduced is the financial claim's price, that we denote by $\varphi$: 
$$
\varphi(t,q,s) := \mathbb E^{\mathbb Q} \left[ \widetilde S_T^2 \vert q_t =q , S_t = s \right].
$$
\begin{rem}
{\rm Notice that the filtration generated by $q$ and $\widetilde S$ is the same as the one associated with $q$ and $S$, namely $\mathcal F_t^{\widetilde S} \vee \mathcal F_t^q = \mathcal F_t^{S} \vee \mathcal F_t^q $ for every $t \in [0,T]$. We conveniently choose to consider $q$ and $S$ as state variables.}
\end{rem}
We clearly expect $\varphi$ to solve the following PDE
\begin{equation}\label{P_Sstoch}\left\{
\begin{array}{rcl}
\varphi_t + \frac{\sigma^2}{2}  (1+z(t,q,s)) \varphi_{ss} & = & 0 \\
\varphi( T,q,s) & = & {(s-a q)}^2.
\end{array}
\right.
\end{equation}
Since we formally have $\Delta := \frac{\partial \varphi}{ \partial \widetilde s} = \frac{\partial \varphi}{ \partial s} \frac{\partial s}{ \partial \widetilde s} =  \frac{\partial \varphi}{ \partial s}$, using the dynamics $\widetilde S$ under $\mathbb P$ together with \eqref{DeltaH} we find that the value function $v^{\lambda}$ satisfies
\begin{equation*}\label{v_Sstoch}
v^{\lambda}(0,q_0,s_0) = \sup_{u,z} \ \mathbb E_{q_0,s_0} \left [  \int_0^T \left(  q_t (S_t - a q_t) - \frac{g}{2} z_t^2 - \frac{\kappa}{2} u_t^2 - \lambda \varphi_s (t,q_t, S_t) (\mu - a u_t) \right) dt   \right ]. 
\end{equation*}
The value function $v^{\lambda}$ is solution to the following HJB equation, which depends on $\varphi$ (satisfying \eqref{P_Sstoch})
\begin{equation}\label{HJB_Sstoch}
\left\{
\begin{array}{rcl}
\displaystyle - v_t &=& \sup_{u,z} \big\{ q(s- a q) - \frac{g}{2} z^2 - \frac{\kappa}{2} u^2 - \lambda \varphi_s (\mu - a u ) + v_q u + v_s \mu \\
& & + \frac{ \sigma^2}{2} (1+z) v_{ss} \big\} \\
v( T,q,s) & = & 0
\end{array}
\right.
\end{equation}

\begin{ans}\label{ans-case2}
We guess the value function $v$ and the price $\varphi$ have the following form
\begin{eqnarray*}
v(t,q,s) &=&  A(t) q^2 + B(t) s^2 + C(t) qs + D(t) q + E(t) s + F(t), \\
\varphi(t,q,s) &=&  \bar A(t) q^2 + \bar B(t) s^2 + \bar C(t) qs + \bar D(t) q + \bar E(t) s + \bar F(t).
\end{eqnarray*}
\end{ans}

The first order conditions on the \eqref{HJB_Sstoch} lead us to the candidate optimal controls
\begin{equation*}
\widehat z = \frac{\sigma^2}{2 g} v_{ss} \, , \qquad   \widehat u = \frac{1}{\kappa} \left( a\lambda \varphi_s + v_q \right),
\end{equation*}
and thus, using the ansatz above,
\begin{equation}\label{hat_uz}
\widehat z_t = \frac{\sigma^2}{g} B(t) \qquad   \widehat u_t = \frac{1}{k} \left\{  \lambda  a  \left[ 2 \bar B(t) s + \bar C(t) q + \bar E(t)  \right] + 2 A(t) q + C(t) s + D(t)  \right\},
\end{equation}
for all $t \in [0,T]$.

\paragraph{Solving for $\varphi$.}
We can now explicitly find $\varphi$ by exploiting the Ansatz \ref{ans-case2} and replacing the control pair $(\widehat u, \widehat z)$ into \eqref{P_Sstoch}. We find:
\begin{equation*}
\left\{
\begin{array}{rcl}
\bar A(t) & \equiv & a^2\\
\bar B(t) & \equiv & 1 \\
\bar C(t) & \equiv & - 2 a \\
\bar D(t) & \equiv & 0 \\
\bar E(t) & \equiv & 0 \\
\bar F(t) & = & \displaystyle \int_t^T \sigma^2 \left[  1+\frac{\sigma^2}{g} B(u) \right] du \\
\end{array}
\right.
\end{equation*}

\paragraph{Solving for $v^{\lambda}$.}
We proceed in the same way as above to find the value function $v^{\lambda}$. We find the following systems of ODEs
\begin{align} \label{riccati-case2}
A'(t) & =  a - \frac{2}{\kappa} {(A(t) - \lambda a^2)}^2 \\
B'(t) & =  -   \frac{1}{2\kappa} {(2 \lambda a + C(t))}^2 \\
C'(t) & = -1-\frac{2}{\kappa} (2 \lambda a + C(t))  (A(t) - \lambda a^2)\\
D'(t) & =     - \frac{2}{k} D(t) \big( A(t) - a^2 \lambda \big) - \mu \big( C(t) + 2 a \lambda\big) \\
E'(t) & =   - \frac{1}{k} D(t) \big(C(t) + 2 a \lambda \big) - 2 \mu ( B(t) - \lambda)  \\
F'(t) & =   -\frac{1}{2} \frac{\sigma^4 B(t)^2}{g} - \frac{D(t)^2}{2\kappa} - \mu E(t) - \sigma^2 B(t)  \label{F-2}
\end{align}
with null terminal conditions $A(T)=\cdots = F(T)=0$.

\begin{rem} \label{rem:model2}
{\rm Notice, in particular, that $B$ is a positive decreasing function of time, which implies (recall Equation \eqref{hat_uz}) that the control $\widehat z$ is always positive. This means that there is always interest in increasing the market price volatility, even if the producer buys the derivative.}
\end{rem}

\subsection{Verification}
\begin{Theorem}\label{verif-case2}
Let $A,B,C,D$ be solutions to the system \eqref{riccati-case2}-\eqref{F-2}. There exists an optimal policy $(\widehat u, \widehat z) \in \mathcal A$ for problem \eqref{maxJ-case2}, where 
\begin{equation}
\label{opt-uz-case2} 
\widehat u_t = \frac{1}{\kappa}\left[ (2\lambda a + C(t)) \widehat S_t  + 2(A(t) -\lambda a^2) \widehat q_t + D(t) \right], \quad \widehat z_t = \frac{\sigma^2}{g}B(t), \quad t\in [0,T].\end{equation}
The no-arbitrage price process for the derivative $h_T$ is given by
\begin{equation}\label{h-price-case2} 
\widehat h_t := h_t^{\widehat u, \widehat z}= (\widehat S_t -a\widehat q_t)^2 + \sigma^2  \int_t ^T \left( 1+ \frac{\sigma^2}{g}B(u)\right) du, \quad t \in [0,T],
\end{equation}
and the hedging process is
\begin{equation}
\widehat \Delta _t := \Delta_t^{\widehat u, \widehat z} = 2 (\widehat S_t -a\widehat q_t), \quad t\in [0,T].
\end{equation}
\end{Theorem}

\begin{proof}
The proof is very similar to the previous model (cf. Theorem \ref{verif-case1}), hence we give details only for those steps which are slightly different.
\begin{enumerate}[wide, labelwidth=!, labelindent=0pt] 
\item \emph{Admissibility.}
First, we observe that checking the admissibility property (i), namely  $\widehat z_t > -1$ for all $t\in [0,T]$, is actually easier as no conditions on small $T$ are needed. Indeed, $\widehat z_t > -1$ if and only if $B(t) > - g/\sigma^2$, where $B$ solves the corresponding equation in the system \eqref{riccati-case2}-\eqref{F-2}. The latter inequality is satisfied since, being $B'(t) \le 0$ and $B(T) =0$, we have $B(t) \ge 0$ for all $t\in [0,T]$.

Regarding the integrability properties in (i), checking them is equivalent to show that
\[ \mathbb E\left[\int_0 ^T \widehat q_t ^2 dt\right] < \infty, \quad \mathbb E\left[\int_0 ^T \widehat S_t ^2 dt\right] < \infty.\]
As for the condition on $\widehat S$, observe that under $\mathbb P$ the process $\widehat S$ satisfies 
$
d \widehat S_t = \mu dt + \sigma \sqrt{1 + \frac{\sigma^2}{g} B(t)} dW_t,
$
namely it is a Gaussian process. So, as previously noticed in the proof of Theorem \ref{verif-case1}, we apply first of all  Fubini's theorem and then we remark that the function $t \mapsto \mathbb E[\widehat S_t ^2]$ is continuous over $[0,T]$ and its integral is finite.  
We now work on $\mathbb E[\widehat q_t ^2]$, which is more delicate, since $d \widehat q_t = \widehat  u_t dt$ and $\widehat u$ depends also on $\widehat S$ (see \eqref{opt-uz-case2}). We have 
$$
\widehat q_t = q_0 +\frac{1}{\kappa} \int_0^t (2\lambda a + C(s))  \widehat S_s  ds + \frac{2}{\kappa} \int_0^t (A(s) -\lambda a^2) \widehat q_s ds + \frac{1}{\kappa} \int_0^t D(s) ds  ,
$$
and so
\begin{eqnarray*}
\mathbb E[\widehat q_t ^2] & \le & 3 \left( q_0 + \frac{1}{\kappa} \int_0^t D(s) ds \right)^2 +  \frac{3}{\kappa^2} \mathbb E \left( \int_0^t (2\lambda a + C(s))  \widehat S_s  ds \right)^2 \\
&& +  \frac{12}{\kappa^2} \mathbb E \left( \int_0^t (A(s) -\lambda a^2) \widehat q_s ds \right)^2   \\
& \le & \underbrace{3 \left( q_0 + \frac{1}{\kappa} \int_0^t D(s) ds \right)^2 +  \frac{3  t }{\kappa^2} \int_0^t [2\lambda a + C(s)]^2 \mathbb E[ \widehat S_s^2]  ds }_{:= \alpha(t)} \\
&& +  \frac{12 t}{\kappa^2}  \int_0^t [A(s) -\lambda a^2]^2 \mathbb E[\widehat q_s^2] ds  .
\end{eqnarray*}
Now, $\mathbb E[ \widehat S_s^2]$ is positive and finite and so we can safely introduce the positive continuous function $\alpha$ as above and write
\begin{eqnarray*}
\mathbb E[\widehat q_t ^2] & \le & \alpha(t) +  \frac{12 t}{\kappa^2}  \int_0^t [A(s) -\lambda a^2]^2 \mathbb E[\widehat q_s^2] ds  .
\end{eqnarray*}
An application of Gronwall's lemma leads to
$$
\mathbb E[\widehat q_t ^2] \le \alpha(t) +\int_0^t \alpha(s) K(s) e^{\int_s^t K(u) du} ds,
$$
with $K(u) := \frac{12 t}{\kappa^2} [A(u) -\lambda a^2]^2 >0$. So, $t \rightarrow\mathbb E[\widehat q_t ^2] $ is bounded by a continuous function and its integral over $[0,T]$ is finite.

\noindent Regarding condition (ii), we proceed again as in the proof of Theorem \ref{verif-case1} by checking Assumption 2.2 in \cite{Rydberg97}, which in our case is satisfied as long as $\sigma^2 (1+\widehat z_t) > 0$ for all $t\in [0,T]$. This is automatically true, since here $B(t) \ge 0 $ for all $t \in [0,T]$ and so $\sigma^2 (1 + \widehat z_t) =\sigma^2 ( 1 + \frac{\sigma^2}{g} B(t))  \ge \sigma^2 >0$.
To end checking condition (ii), we have to show $\widehat h_T \in L^1(\mathbb P) \cap L^1(\widehat \Q)$. Now, $\widehat h_T = (\widehat S_T -a\widehat q_T)^2$, hence quadratic in both $\widehat S_t$ and in $\widehat q_T$. For what we have seen up to now $q_T ^2$ and $ \widehat S_T$ belong to $  L^1(\mathbb P)$ and they also belong to $ L^1(\widehat \Q)$, which gives the desired property.

It remains to check (iii). We proceed again as in the previous theorem, except that now $\widehat h_t$ in \eqref{h-price-case2} is a function of both $\widehat S$ and $\widehat q$. First, notice that
\[\widehat h_t = \widehat h_0 + 2 \int_0 ^t  (\widehat S_s - a \widehat q_s) (d \widehat S_s - a d \widehat q_s)= \widehat h_0 + \int_0 ^t \widehat \Delta_s d\widetilde S_s.\] 
Now, since $ \widehat h_T = (s_0 - a q_0)^2 + 2 \int_0^T (\widehat S_s - a \widehat q_s) \sigma \sqrt{1 + \widehat z_s} d \widehat W_s$ we have
\begin{eqnarray*}
\widehat {\mathbb E} [h_T \vert \mathcal F_t] &=& (s_0 - a q_0)^2 + 2 \int_0^t (\widehat S_s - a \widehat q_s) \sigma \sqrt{1 + \widehat z_s} d \widehat W_s\\
& &  + 2 \ \widehat {\mathbb E} \left[  \int_t^T (\widehat S_s - a \widehat q_s) \sigma \sqrt{1 + \widehat z_s} d \widehat W_s \right] = \widehat h_t. 
\end{eqnarray*}
Finally, progressive measurability and integrability of $\widehat \Delta $ in condition (iii) in Definition \ref{adm-case2} can be treated exactly as in the proof of Theorem \ref{verif-case1}, using (i).

\item \emph{Optimality} This can be proved by proceeding exactly as in the proof of Theorem \ref{verif-case1}, by taking into account that now the state variable is two-dimensional.
\end{enumerate}
\vspace{-0.7cm}
\end{proof}

\section{Producer--trader competition}
\label{sec:model3}

In this section, we finally consider a game between a producer who can manipulate the commodity price through the drift and a trader who can manipulate the volatility of the price. Hence, the trader is no longer price-taker or passive as in the previous two models. Here, she can affect the price of commodity by paying some (quadratic) cost.

The dynamics for $S$ and $q$ are as in \eqref{eq:sq}. The controls are still given by $(u,z)$ with the big difference that now only $u$ is controlled by the producer, while $z$ is controlled by the trader. Clearly, the corresponding costs are allocated accordingly. When the strategy profile for both players is $(u,z)$ and the derivative price is $h_0$, the producer payoff is
\begin{equation}
J_{\rm pr}(u , z, h_0) = \E \left[  \int_0 ^T \left(q_t (S_t -aq_t)  - \frac{\kappa}{2}u_t ^2 \right)dt + \lambda (h_0 -h_T) \right].
\end{equation}
On the other hand, the trader who has the opposite position in the derivatives will get the payoff
\begin{equation}
J_{\rm tr} (u,z,h_0) = \E\left[ - \frac{g}{2} \int_0 ^T z_t ^2 dt - \lambda (h_0 -h_T) \right].
\end{equation}
Now, we can give the definition of admissible policies, including the strategy profiles of both players together with the pricing and hedging strategy of the investment department in the production firm. Notice that it is the same as for the model in Section \ref{sec:model2}. We recall that $\tilde S_t = S_t -aq_t$, for $t \in [0,T]$.

 \begin{Definition} \label{adm-case3}
We say that any pair $(u,z)$ is \emph{admissible} if the following properties are satisfied:
		\begin{enumerate}
\item[(i)] $(u_t ,z_t )_{t \in [0,T]}$ are progressively measurable processes with values in $\mathbb R \times (-1,\infty)$ such that
\[ \mathbb E\left[\int_0 ^T (u_t ^2 + z^2 _t)dt\right] < \infty, \quad \mathbb E\left[\int_0 ^T \tilde S_t ^2 (1 + z_t)dt\right] < \infty; \]
\item[(ii)] there exists a unique equivalent martingale measure $\mathbb Q^{u,z}$ for the price process $\tilde S$, with $h_T \in L^1 (\mathbb P) \cap L^1 (\mathbb Q^{u,z})$; 
\item[(iii)] there exists a real-valued progressively measurable process $(\Delta_t^{u,z})_{t\in [0,T]}$ satisfying
\[ 
\mathbb E\left[ \int_0^T \left( |\Delta_t^{u,z} u_t| + |\Delta_t^{u,z} |^2 (1+z_t) \right) dt \right] < \infty ,
\]
and such that the following holds $\mathbb Q^{u,z}$-a.s.
\[ h_t^{u,z} := \E^{\Q^{u,z}} [ h_T| \mathcal F_t] = \E^{\Q^{u,z}} [ h_T]   + \int_0 ^t \Delta_s^{u,z} dq_s , \] for all $t \in [0,T]$.\end{enumerate}
The set of all admissible pairs will be denoted by $\mathcal A$. 
\end{Definition}

On the other hand the definition of solution is slightly different, due to the fact that the trader can also play strategically in this model. Before proceeding, we exploit the definition of admissibility, conditions (ii) and (iii) in particular, to rewrite as in the previous two models the payoffs of both the trader and the producer as follows
\begin{eqnarray*}
J_{\rm pr}(u , z, h_0) &=& \E \left[  \int_0 ^T \left(q_t (S_t -aq_t)  - \frac{\kappa}{2}u_t ^2 -\lambda (\mu-au_t) \Delta_t^{u,z} \right)dt \right] =: \tilde J_{\rm pr}(u , z),\\
J_{\rm tr} (u,z,h_0) &=& \E\left[ \int_0 ^T \left(- \frac{g}{2}  z_t ^2 +\lambda (\mu-au_t) \Delta_t^{u,z} \right) dt \right] =: \tilde J_{\rm tr} (u,z),
\end{eqnarray*}
for any admissible pair $(u,z) \in \mathcal A$.

\begin{Definition}\label{def:eq-case3}
We say that the pair $(\widehat u, \widehat z) \in \mathcal A$ is a Nash equilibrium if it is a Nash equilibrium between the producer and the trader, i.e.,
\begin{equation}\label{nash}
\tilde J_{\rm pr}(\widehat u , \widehat z) \ge \tilde J_{\rm pr}(u , \widehat z), \quad \tilde J_{\rm tr}(\widehat u , \widehat z) \ge \tilde J_{\rm tr}(\widehat u , z),
\end{equation}
for all deviations $u,z$ such that $(u,\widehat z)$ and $(\widehat u,z)$ belong to $\mathcal A$.
\end{Definition}

\subsection{Heuristics} 
We want to compute explicitly a Nash equilibrium for the game between producer and trader, described just above. We start from some heuristics that would lead to some candidate equilibrium, while the rigorous verification is postponed to the next sub-section as usual. 

First, assuming $h_t = \varphi(t,q_t,S_t)$ and consequently $\Delta_t = \varphi_s (t,q_t,S_t)$, while exploiting the market completeness as for the previous two models, we can rewrite the running best-response functions of the two players as follows
\begin{eqnarray*}
v(t,q,s; z) &=& \sup_u \E \left[  \int_t ^T \left(q_r (S_r -aq_r)  - \frac{\kappa}{2}u_r ^2 - \lambda (\mu -a u_r) \varphi_s \right)dr \mid q_t =q, S_t =s  \right], \\
w(t,q,s; u) &=& \sup_z \E\left[  \int_t ^T \left( - \frac{g}{2} z_r ^2 + \lambda(\mu -au_r) \varphi_s \right) dr \mid q_t =q, S_t =s \right],
\end{eqnarray*}
where $\varphi_s$ is the delta hedging of the derivative, that will have to be determined at the equilibrium. It is reasonable to expect that the derivative's price $\varphi(t,q,s)$ is the solution to the following PDE:
\begin{equation} \varphi_t + \frac{\sigma^2}{2}(1+z) \varphi_{ss} =0, \quad \varphi(T, s,q) = (s-aq)^2, \label{PDE-phi}\end{equation}
for some function $z=z(t,q,s)$ coming from the trader's best-response.
The PDE above is coupled with the following two HJB equations, arising from the best-response functions of the two players:  
\begin{eqnarray}
-v_t &=& \sup_{u} \left\{ q(s- a q)  - \frac{\kappa}{2} u^2 - \lambda \varphi_s (\mu - a u ) + v_q u + v_s \mu + \frac{\sigma^2}{2} (1+z) v_{ss} \right\}, \label{HJB-prod} \\
-w_t &=& \sup_{z} \left\{  - \frac{g}{2} z^2 + \lambda \varphi_s (\mu - a u ) + w_q u + w_s \mu + \frac{\sigma^2}{2} (1+z) w_{ss} \right\}, \label{HJB-trad}
\end{eqnarray}
with terminal conditions $v(T)=w(T)=0$.
Solving the optimization problems within the HJB equations above, we get the (candidate) equilibrium strategies for the producer and the trader in terms of the corresponding payoff functions:
\b*
\widehat u = \frac{1}{\kappa} \big( v_q + \lambda a \varphi_s \big),  \quad \widehat z = \frac{\sigma^2}{2 g} w_{ss}.
\e*
The HJB equations for the producer and the trader become respectively as
\b*
- v_t = q (s - a q) + \frac{1}{2 \kappa} \big( v_q + \lambda a \varphi_s \big)^2 
         - \mu \lambda \varphi_s + \mu v_s 
         + \frac{\sigma^2}{2} \left( 1 + \frac{\sigma^2}{2 g} w_{ss} \right) v_{ss}
\e*
and
\b*
- w_t = \lambda \mu \varphi_s - \frac{1}{\kappa} \big(  v_q + \lambda a \varphi_s \big)  \big(   \lambda a \varphi_s - w_q \big)
        + \mu w_s +  \frac{\sigma^2}{2} w_{ss} +  \frac{\sigma^4}{8 g} w^2_{ss} .
\e*
Furthermore, the PDE giving the option equilibrium price $\varphi$ is given:
\b*
\varphi_t + \frac{\sigma^2}{2} \left( 1 + \frac{\sigma^2}{2g} w_{ss} \right) \varphi_{ss} = 0, \text{with} \quad \varphi(T,q,s) = (s-aq)^2.
\e*
Analogously as in the previous two models, we use the following ansatz for $w$:
 \b*
w(t,q,s) = A_w(t) q^2 + B_w(t) s^2 + C_w(t) q s + D_w(t) q + E_w(t) s + F_w(t),
 \e*
 and similarly for $v$ and $\varphi$ with self-explanatory notation for their coefficients. Therefore, using the ansatz and proceeding in the usual way, we easily get
\b*
\varphi(t,q,s) = (s - a q) ^2 + F_{\varphi}(t),
\e*
where
\b*
F_{\varphi}(t) := \sigma^2 \int_t^T \left( 1 + \frac{\sigma^2}{g} B_{w}(r) \right) dr.
\e*
After tedious yet straightforward computations we obtain
\begin{align}
- A_v^{'} & = - a + \frac{2}{\kappa} \left(   A_v  - a^2 \lambda  \right)^2 \label{eqAw} \\
- B_v^{'} & = \frac{1}{2 \kappa} \left(   C_v  + 2 a \lambda   \right)^2  \\
- C_v^{'} & = 1 + \frac{2}{\kappa} \left(  A_v  -  a^2 \lambda \right)  \left(   C_v  + 2 a \lambda   \right) \\
- D_v^{'} & =  \mu \left(   C_v  + 2 a \lambda   \right)  +  \frac{2}{\kappa} D_v \left(   A_v  -  a^2 \lambda \right)  \\
- E_v^{'} & =  2 \mu \left(   B_v  -  \lambda   \right)  +  \frac{1}{\kappa} D_v \left(   C_v  + 2 a \lambda \right)  \\
- F_v^{'} & =   \mu   E_v    +  \frac{1}{2 \kappa} D^2_v + \sigma^2 \left( 1 + \frac{\sigma^2}{g}  B_w \right) B_v  \\
- A_w^{'} & = \frac{4 }{\kappa} \left(   A_v  -  a^2 \lambda \right)   \left( A_w  +  a^2 \lambda \right) \\
- B_w^{'} & =  \frac{1}{\kappa}  \left( C_w - 2 a \lambda \right) \left( C_v + 2 \lambda a \right)  \label{eq-Bw}  \\
- C_w^{'} & =    \frac{2 }{\kappa}  \left[ ( A_v - \lambda a^2) ( C_w - 2 a \lambda ) 
                  +  ( A_w  +  a^2 \lambda ) ( C_v + 2 \lambda a )  \right]  \\
- D_w^{'} & =    \mu \left( C_w - 2 a \lambda \right) 
                  + \frac{2  }{\kappa} D_v \left( A_w  +  a^2 \lambda \right)  +  \frac{2  }{\kappa} D_w \left( A_v  -  a^2 \lambda \right)  \\
- E_w^{'} & =    2 \mu \left( B_w +  \lambda \right)
                  + \frac{1 }{\kappa} D_w \left( C_v  + 2  a \lambda \right)  +  \frac{1  }{\kappa} D_v \left( C_w  - 2 a \lambda \right)  \\
- F_w^{'} & =   \mu   E_w    +  \frac{1}{\kappa} D_v D_w + \sigma^2 B_w +  \frac{\sigma^4}{2 g}  B^2_w
\label{eqFw}
\end{align}
with zero terminal condition for all ODEs above. The first equation, which is a Riccati ODE, can be solved explicitly  giving the same expression as for $D$ in the first model:
\[ A_v (t) =  - \frac{2(a-\frac{2\lambda^2 a^4}{\kappa})(e^{\theta (T-t)} -1)}{\theta (e^{\theta(T-t)} +1) +\frac{4\lambda a^2}{\kappa} (e^{\theta(T-t)} -1)}, \quad \theta := \sqrt{\frac{8a}{\kappa}}.\]
The other equations are linear, hence they can be solved in integral form. For the moment, we give only the expressions for the coefficients that we need in order to compute the equilibrium strategy $\widehat z$ of the trader. They are given by:
\begin{eqnarray*}
A_w (t) &=& \frac{4a^2 \lambda}{\kappa} \int_t ^T e^{\frac{4}{\kappa}\int_t ^u (A_v(r) - a^2 \lambda) dr}  (A_v(u) - a^2 \lambda) du,\\
B_w (t) &=& \frac{1}{\kappa} \int_t ^T (C_w(r)-2a\lambda)( C_v(r) + 2a\lambda)dr ,\\
C_w(t) &=& \frac{2}{\kappa} \int_t ^T e^{\frac{2}{\kappa}\int_t ^u (A_v(r) - a^2 \lambda) dr} \left[ (A_w(u) + a^2 \lambda) (C_v (u) + 2a\lambda) - 2a\lambda (A_v (u) - a^2 \lambda) \right] du ,\\
C_v (t) &=& \int_t ^T e^{\frac{2}{\kappa}\int_t ^u (A_v(r) - a^2 \lambda) dr} \left( 1+ \frac{4 a \lambda}{\kappa}(A_v (u) - a^2 \lambda ) \right) du.
\end{eqnarray*}

\subsection{Verification}
\begin{Theorem}\label{verif-case3}
Let $A_v,C_v,D_v, B_w$ be solutions to the system \eqref{eqAw}-\eqref{eqFw} such that 
\begin{equation}
\label{cond-Bw} B_w (t) > -\frac{g}{\sigma^2}, \quad t \in [0,T].
\end{equation}
Then there exists a Nash equilibrium $(\widehat u, \widehat z) \in \mathcal A$ as in Definition \ref{def:eq-case3}, where
\begin{equation}
\label{opt-uz-case3} \widehat u_t = \frac{1}{\kappa}\left[ C_v (t) \widehat S_t  + 2A_v (t) \widehat q_t + D_v (t) + 2\lambda a (\widehat S_t -a\widehat q_t) \right], \quad \widehat z_t = \frac{\sigma^2}{g}B_w(t), \quad t\in [0,T].
\end{equation}
The no-arbitrage equilibrium price process for the derivative $h_T$ is given by
\begin{equation}\label{h-price-case3} 
\widehat h_t := h_t^{\widehat u, \widehat z} = (\widehat S_t -a\widehat q_t)^2 + \sigma^2 \int_t ^T \left( 1+ \frac{\sigma^2}{g}B_w(u)\right) du, \quad t \in [0,T],
\end{equation}
and the hedging process at equilibrium is
\begin{equation}
\widehat \Delta _t := \Delta_t^{\widehat u, \widehat z} = 2 (\widehat S_t -a\widehat q_t), \quad t\in [0,T].
\end{equation}
\end{Theorem}

\begin{proof}
Analogously as for the previous two models, the proof is structured in two steps.\begin{enumerate}[wide, labelwidth=!, labelindent=0pt]
\item \emph{Admissibility}. We start from the verification that the pair $(\widehat u, \widehat z)$ belongs to $\mathcal A$. The two processes $\widehat u,\widehat z$ are clearly progressively measurable by definition, where $\widehat u$ takes real values. Moreover, assumption \eqref{cond-Bw} implies that $\widehat z_t > -1$ for all $t \in [0,T]$. Concerning the integrability properties in Definition \ref{adm-case3}(i), since both $\widehat u$ and $\widehat z$ are affine in the state variables $\widehat q_t,\widehat S_t$ with time continuous (hence bounded) coefficients, they boil down to checking
\begin{equation}\label{square-int-qS} \E \left[ \int_0 ^T (\widehat q_t ^2 + \widehat S_t ^2) dt \right] < \infty.\end{equation}
For the square integrability of $\widehat S$, observe that since $\widehat z$ is a deterministic continuous function of time, each $\widehat S_t$ is normally distributed, hence it has every moment and they are continuous in time. Therefore $\E[\int_0 ^T \widehat S_t ^2 dt] < \infty$. To verify the square integrability of $\widehat q$, notice that at equilibrium we have
\[ d\widehat q_t = [\alpha(t)\widehat q_t + \beta(t)\widehat S_t + \gamma(t)] dt, \quad \widehat q_0 =q_0,\]
for some deterministic continuous functions of time $\alpha, \beta$ and $\gamma$. Such a linear ODE can be solved pathwise, giving
\[ \widehat q_t = e^{\int_0 ^t \alpha(r) dr} \left(q_0 + \int_0 ^t e^{-\int_0 ^r \alpha(u)du}\left( \beta(t)\widehat S_r + \gamma(r)\right) dr\right), \quad t\in [0,T]. \] This implies that showing $\E[\int_0 ^T \widehat q_t ^2 dt] < \infty$ reduces to $\E[ \int_0 ^T (\int_0 ^t \widehat S_ r dr)^2 dt] < \infty$, which follows since $\Sigma_t := \int_0 ^t \widehat S_r dr$ is a Gaussian process with time continuous second moment.

\noindent Regarding condition (ii), we need to show that there exists a unique EMM $\widehat{\mathbb Q} = \mathbb Q^{\widehat u,\widehat z}$ for the production process $\widehat q$. 
Let $\widehat \gamma_t := \frac{\mu-a\widehat u_t}{\sigma \sqrt{1+\widehat z_t}}$ and let us consider
\[ L_t ^{\widehat u,\widehat z} := \exp \left\{ -\int_0 ^t \widehat \gamma_t dW_t - \frac{1}{2} \int_0 ^t \widehat \gamma_t ^2 dt  \right\},\quad t \in [0,T].\]
We use once more \cite[Theorem 2.1]{Rydberg97} to prove that under our assumptions the probability $d\widehat \Q := L_T ^{\widehat u,\widehat z} d\mathbb P$ is well-defined (see also \cite{Ruf15}). For this model, Assumption 2.2 in \cite{Rydberg97} is satisfied as long as $\sigma^2 (1+\widehat z_t) > 0$ for all $t\in [0,T]$, which is immediately given by our assumption that $B_w (t) > -g/\sigma^2$ ensuring, as we already saw above, that $\widehat z_t > -1$ for all $t \in [0,T]$. 

\noindent To end this part, we need to check property (iii) in Definition \ref{adm-case3}. The first part, relative to $\widehat h$, is done as in the proof of Theorem \ref{verif-case1}, hence the details are omitted. Concerning $\widehat \Delta$, first notice that it is clearly progressively measurable and takes real values. Due to the fact that $\widehat \Delta_t$ is linear in both $\widehat q_t$ and $\widehat S_t$, its integrability property is equivalent to \eqref{square-int-qS}, which has already been checked before.

\item \emph{Equilibrium}. We are going to use the martingale optimality principle here as well to verify that at the proposed equilibrium $(\widehat u, \widehat z)$ both players are implementing an optimal response to each other strategy. Let us consider the producer first and define the following process
\begin{equation}\label{def-Y-prod}
Y^{u,\widehat z}_t := \int_0 ^t  \left(q_r (\widehat S_r -aq_r) - \frac{\kappa}{2}u_r ^2 - \lambda \Delta_r^{u, \widehat z}  (\mu-a u_r) \right)dr + W(t,q_t, \widehat S_t),
\end{equation}
with
\[ W(t,q,s) = A_w(t) q^2 + B_w(t) s^2 + C_w(t) q s + D_w(t) q + E_w(t) s + F_w(t).\]
Similarly as in the proof of Theorem \ref{verif-case1}, one can verify by applying It\^o's formula and the HJB equation \eqref{HJB-prod} satisfied by the function $W(t,q,s)$ by construction in the heuristics part, that $Y^{u,\widehat z}_t$ is a supermartingale for all $u$ such that $(u, \widehat z) \in \mathcal A$, and a martingale for $u=\widehat u$. 
For the trader, we consider the process
\begin{equation}\label{def-Y}
Z^{\widehat u, z}_t := \int_0 ^t  \left( - \frac{g}{2}z_r ^2 + \lambda \Delta_r^{\widehat u, z} (\mu-a \widehat u_r) \right)dr + U(t, \widehat q_t, S_t),
\end{equation}
with
\[ U(t,q,s) = A_v(t) q^2 + B_v(t) s^2 + C_v(t) q s + D_v(t) q + E_v(t) s + F_v(t).\]
Applying the same arguments as for the producer, one can easily check that $Z^{\widehat u, z}$ is a supermartingale for all $z$ such that $(\widehat u, z) \in \mathcal A$ and a martingale for $z=\widehat z$. 
\end{enumerate}
Finally, an application of the martingale optimality principle combined with the admissibility of $(\widehat u, \widehat z)$, gives that the latter is a Nash equilibrium as in Definition \ref{def:eq-case3}. Therefore, the proof is complete.
\end{proof}

\begin{rem}
{\rm Observe that in the theorem above we assumed that $B_w (t) > -g/\sigma^2$ for all $t \in [0,T]$. This is satisfied when the maturity $T$ is small enough. Indeed, one can reason heuristically in the following way: when $T \approx 0$, using the equation \eqref{eq-Bw} for $B_w$ we have $B'_w (T) \approx \frac{1}{\kappa}(2\lambda a)^2 >0$ and we also have $B_w (T) = 0$. Therefore, it is natural to expect $B_w (t) > -g/\sigma^2$ for all $t \in [0,T]$ when $T$ is small enough, which would also guarantee that the Radon-Nikodym derivative $d\mathbb Q^{\widehat u,\widehat z}/d\mathbb P = L_T ^{\widehat u,\widehat z}$ is well-defined (see the second part of the proof above). Unfortunately, the study of the function $B_w$ is much more difficult in this case than in the model of Section \ref{sec:model1}, where we were able to quantify precisely how small $T$ must be.}
\end{rem}

\section{Numerical illustration}
\label{sec:numeric}

In this section, we use numerical simulations to illustrate and explain the behaviours of the producer in the three models and of the trader in the third one. We set the drift of the commodity price to zero, $\mu=0$, in Models 2 and 3, to simplify the analysis.

\paragraph{Model 1: production-based manipulation.}The understanding of the first model is quite straightforward and it is illustrated by the first column of Figure~\ref{fig:sim}. Starting from a zero production rate $q$, the optimal strategy of the producer, whether or not she holds a derivative position, is to reach as fast as possible the optimal production rate that maximises the running profit $\displaystyle q^\star := \frac{s_0}{2 a}$. We have seen that when the producer has no position in the derivative market, she has no interest in increasing the volatility using some randomisation of her production rate $q$.  Indeed, her expected profit is proportional to $\E\big[-q_t^2]$ and thus, increasing the volatility decreases her expected profit. On the contrary, since controlling the volatility has a cost, she makes costly efforts to reduce it. When the producer holds a derivative position, we first note that she uses her market power to drive the price of the commodity at maturity to a level that suits her profit. If she has bought  (resp. sold) the derivative, she drives the commodity price up (resp. down). For instance, in case of a sale, the derivative is sold at, say, 100, but at maturity its payoff is close to zero, ensuring a profit of nearly 100. Figure~\ref{fig:v0w0} (left) gives the value function of the producer at time zero as a function of the derivative position. We see that selling derivatives ($\lambda>0$) can only make her better off while buying derivatives requires a certain amount of sales before it is worth the cost.

Regarding the volatility, since the price $h_0$ of the derivative is an increasing function of the realised volatility, the producer may have an interest in increasing the volatility to push the value of the derivative up. But, since increasing the volatility has a negative effect on the expected profit, the producer has to assess this trade-off. Using the Remark \ref{rem:sign-D} together with the expression for $\widehat z$ in (\ref{opt-uz}) and noting that it makes sense to increase the volatility only in case the producer has sold the derivative ($\lambda >0$) we have that
\begin{align}
 \widehat z \ge 0 \quad \Leftrightarrow \quad  \lambda \geq \sqrt{\frac{\kappa}{2 a^3}}.
\end{align}
If the net position exceeds the threshold above, the benefit of increasing the volatility outweighs the cost. The higher the market power, the lower the threshold. Besides, it is worth noting that this threshold does not depend on the cost of intervention $g$ to reduce the volatility. It only depends on the parameters affecting the drift of the commodity price process. In Figure~\ref{fig:sim}, we choose a large short position of $\lambda =1$ which makes the profit on the derivative as important as the profit from production. In that case, the producer increases more than by half the volatility. Figure~\ref{fig:h0} illustrates the variation of the price of the derivative $h_0$, of the expected payoff $\E^\P[ h_T]$ but also of the price of the derivative if no volatility manipulation was undertaken, noted $h_0^{z=0}$, as a function of the holding position $\lambda$. In these simulations, we started the initial production rate at its optimal stationary level $q^\star$ to get rid of transitory effects. For the first model, we observe that $h_0$ is an increasing function of the position, but it varies much less than the expected terminal payoff. It means that much of the benefit from holding a derivative position comes from the manipulation of the price at maturity.

\paragraph{Model 2: production and information based manipulation.} The story for the second model is illustrated by the second column of Figure~\ref{fig:sim} and it has many points in common with the first model: the producer optimal production strategy is to reach the stationary optimal level $q^\star$ and to drive the price at maturity up in case of a purchase and down in case of a sale. But, now, contrary to the former case, as pointed out in Remark~\ref{rem:model2}, the producer always increases the volatility whatever her net position, long or short, because her profit rate is an increasing function of the volatility. As a consequence, we observe in the second column of Figure~\ref{fig:h0} that $h_0$ is always greater than $h_0^{z=0}$. Further, the variation of $\E^\P[h_T] - h_0$ is much larger now, when the producer can separate the manipulation of the drift and of the volatility, than in the first model. Figure~\ref{fig:v0w0} (middle) provides the value function of the producer at time zero as a function of the derivative position. The situation here is very similar to what we observed in the first model, namely selling derivatives ($\lambda>0$) results in a profit, while buying derivatives is worth the cost as soon as the amount of sales exceeds a threshold depending on the model parameters. \medskip

In both models 1 and 2, the capacity of driving the price of the commodity at maturity at a desired level reveals itself an efficient tool to take advantage of a derivative position. If the producer has sold the derivative at, say, 100, she increases her production rate so that at terminal date, the price of the commodity decreases, making the price of the derivative decrease below the initial price and thus ensuring a profit on the derivative.

\paragraph{Model 3: producer-trader competition.} What happens when the producer is facing an opponent who can control the level of volatility? The third column of Figure~\ref{fig:sim} illustrates the interaction between the producer and the trader. The fact that the producer now faces an opponent does not change her overall production strategy: she still reaches the stationary optimal level of production rate $q ^\star$ and she manipulates the commodity price at maturity at her own advantage. But, the actions of the trader on the volatility reduces the potential profit made by the producer on the derivative. When the producer has sold (resp. bought) the derivative to the trader, the trader reduces (resp. increases) the volatility to push the price $h_0$ of the derivative down. The third column of Figure~\ref{fig:h0} shows a much lower variation of the derivative profit $h_0 - \E^\P[h_T]$ than in the first two models. Besides, we observe that for $\lambda>0$, we have $h_0- \E^\P[h_T] >   0$ and for  $\lambda<0$, we have $h_0- \E^\P[h_T] <  0$, meaning that in each case, the trader is making a loss. 

\begin{figure}[hbt!]
\begin{flushleft}
\begin{tabular}{c c c} 
Model 1 & Model 2 & Model 3 \\
\hspace{-5mm}\includegraphics[width=0.33\textwidth]{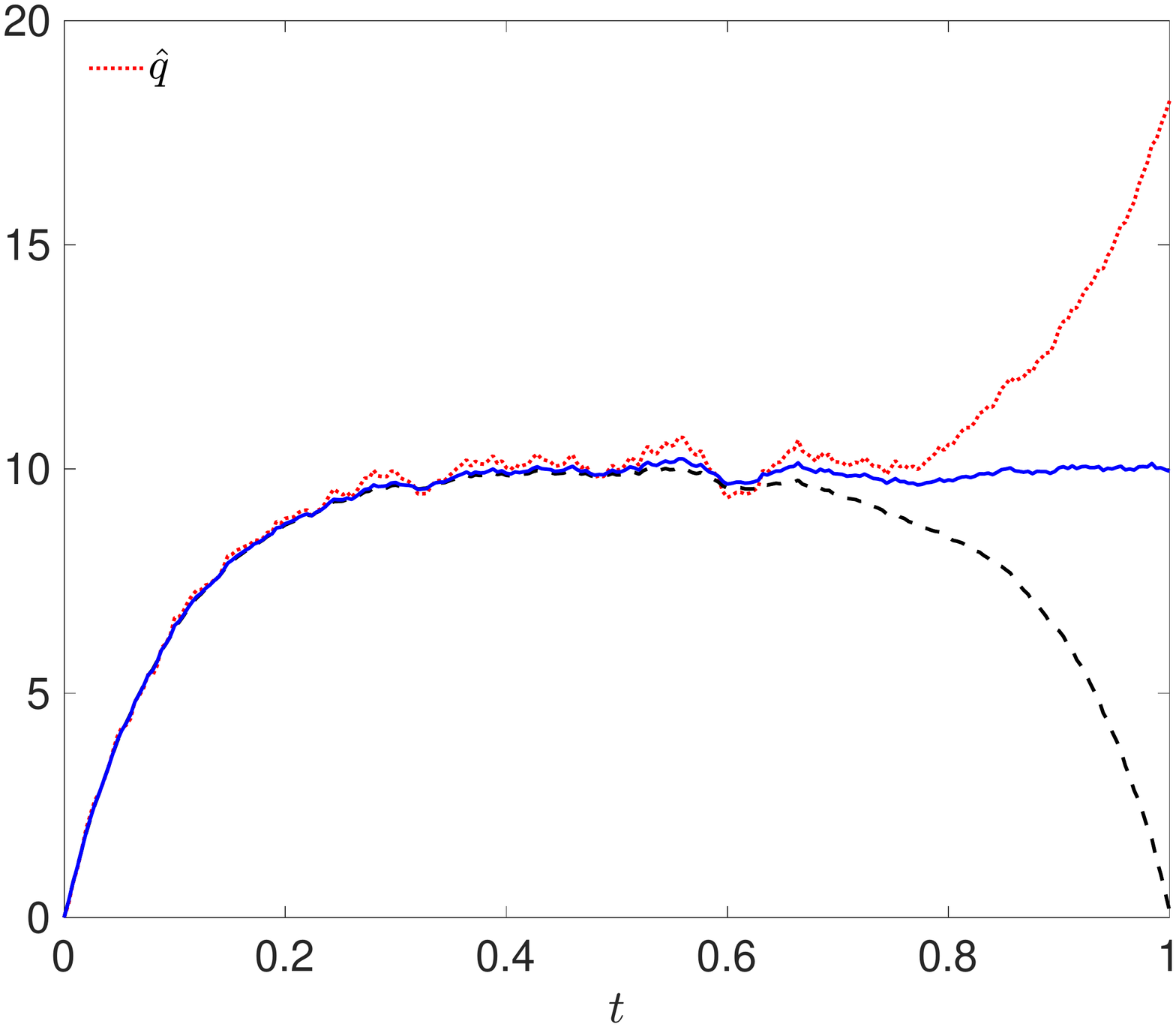} & \includegraphics[width=0.33\textwidth]{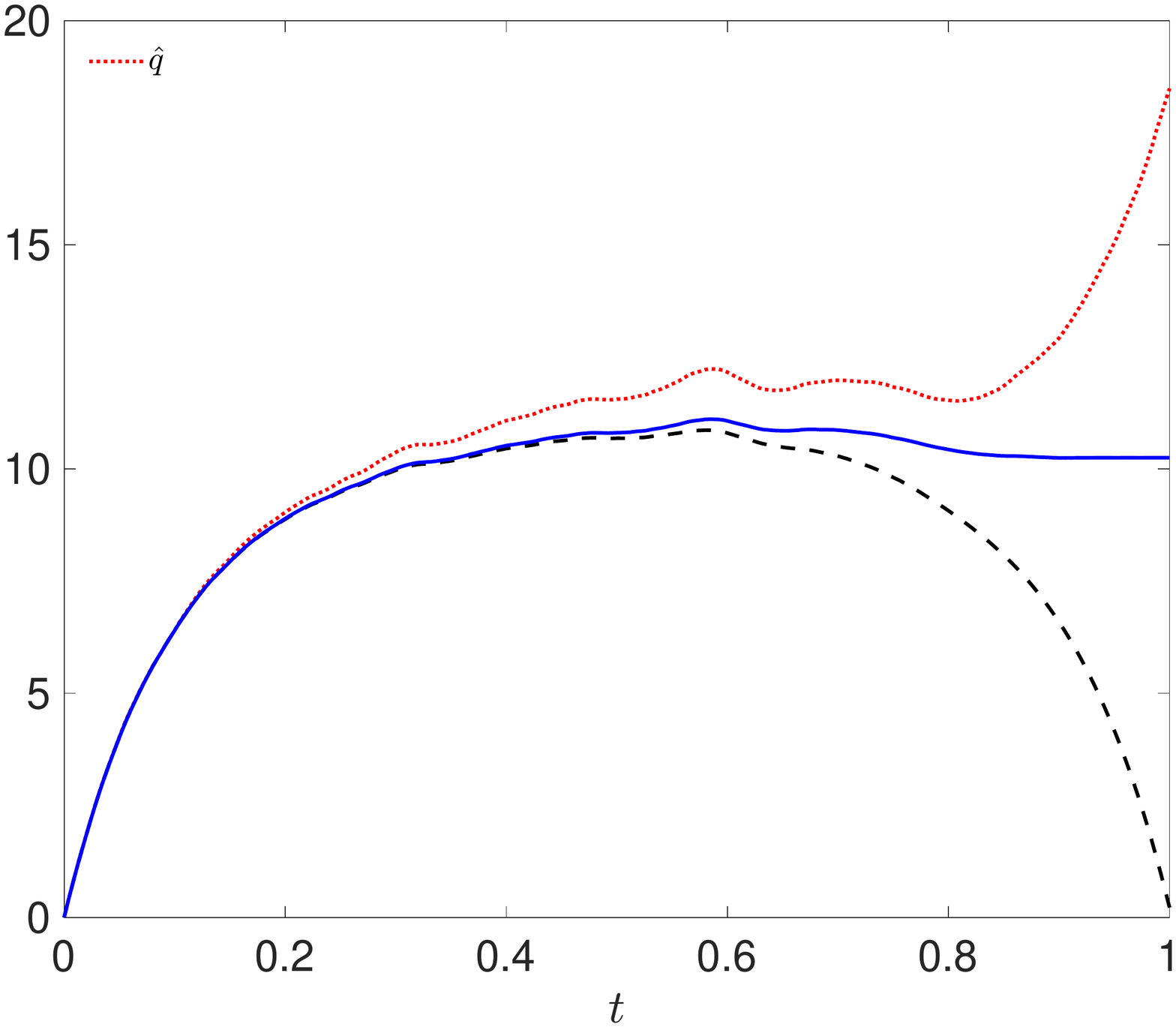} & \includegraphics[width=0.33\textwidth]{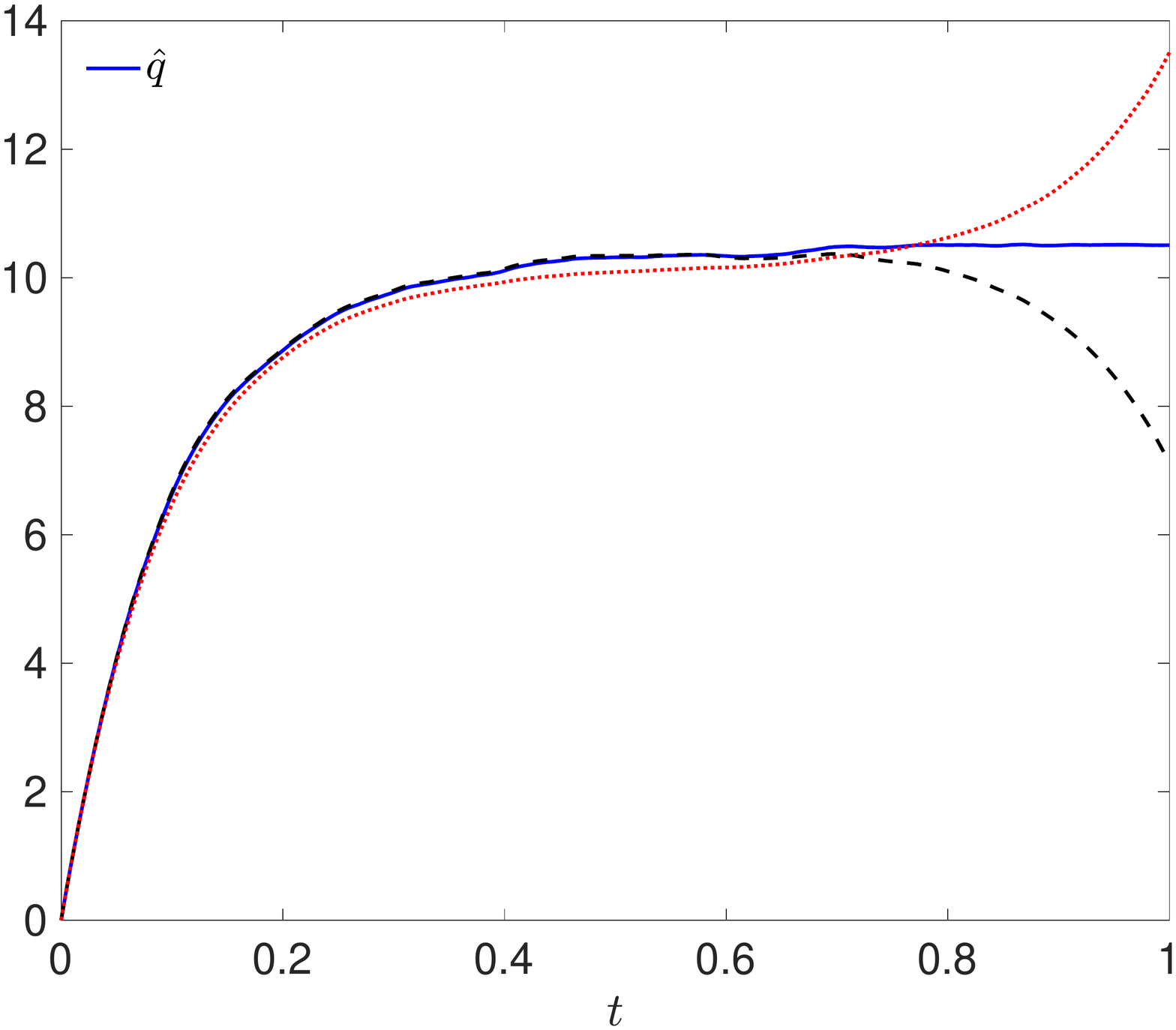}\\
\hspace{-5mm}\includegraphics[width=0.33\textwidth]{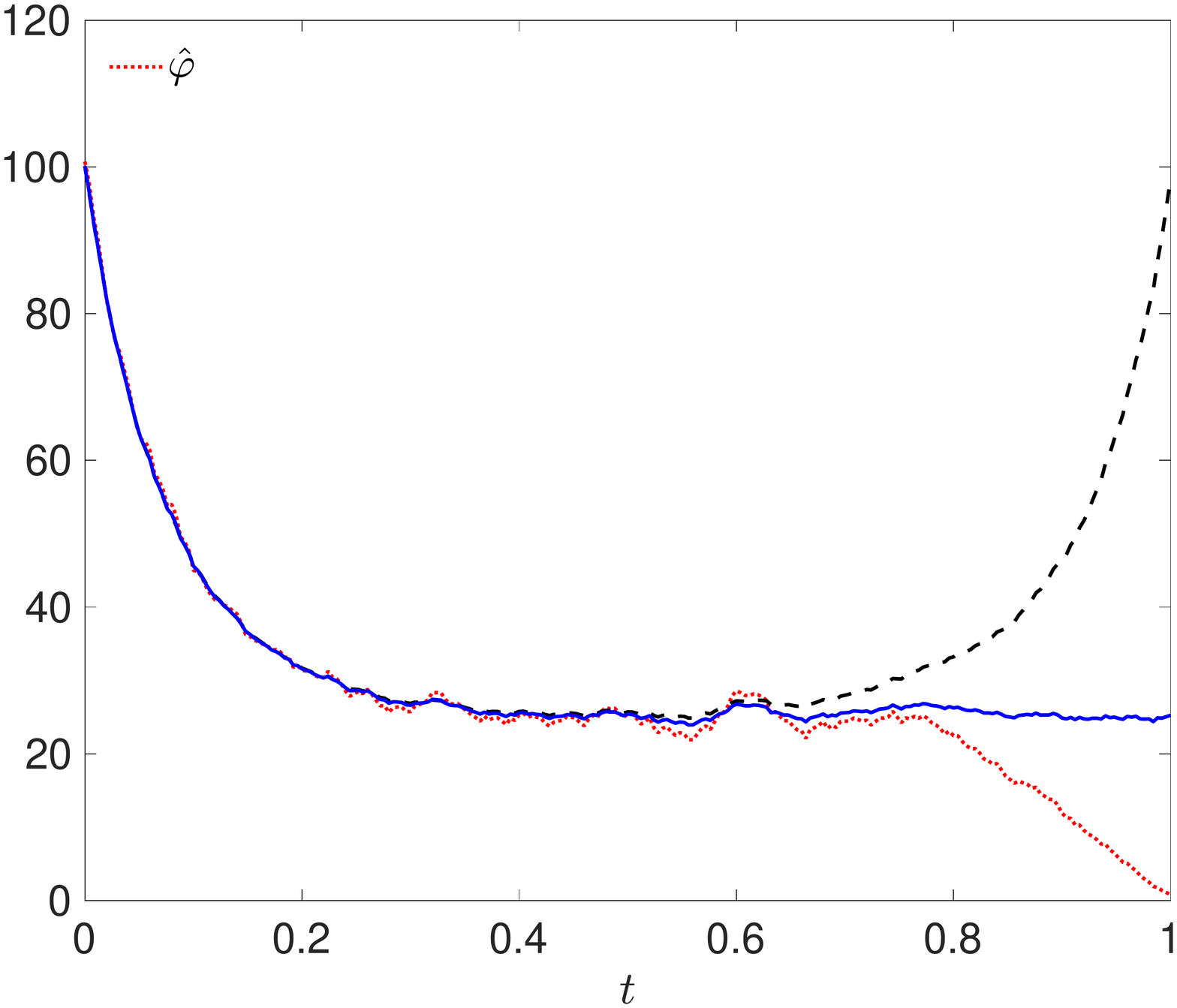} & \includegraphics[width=0.33\textwidth]{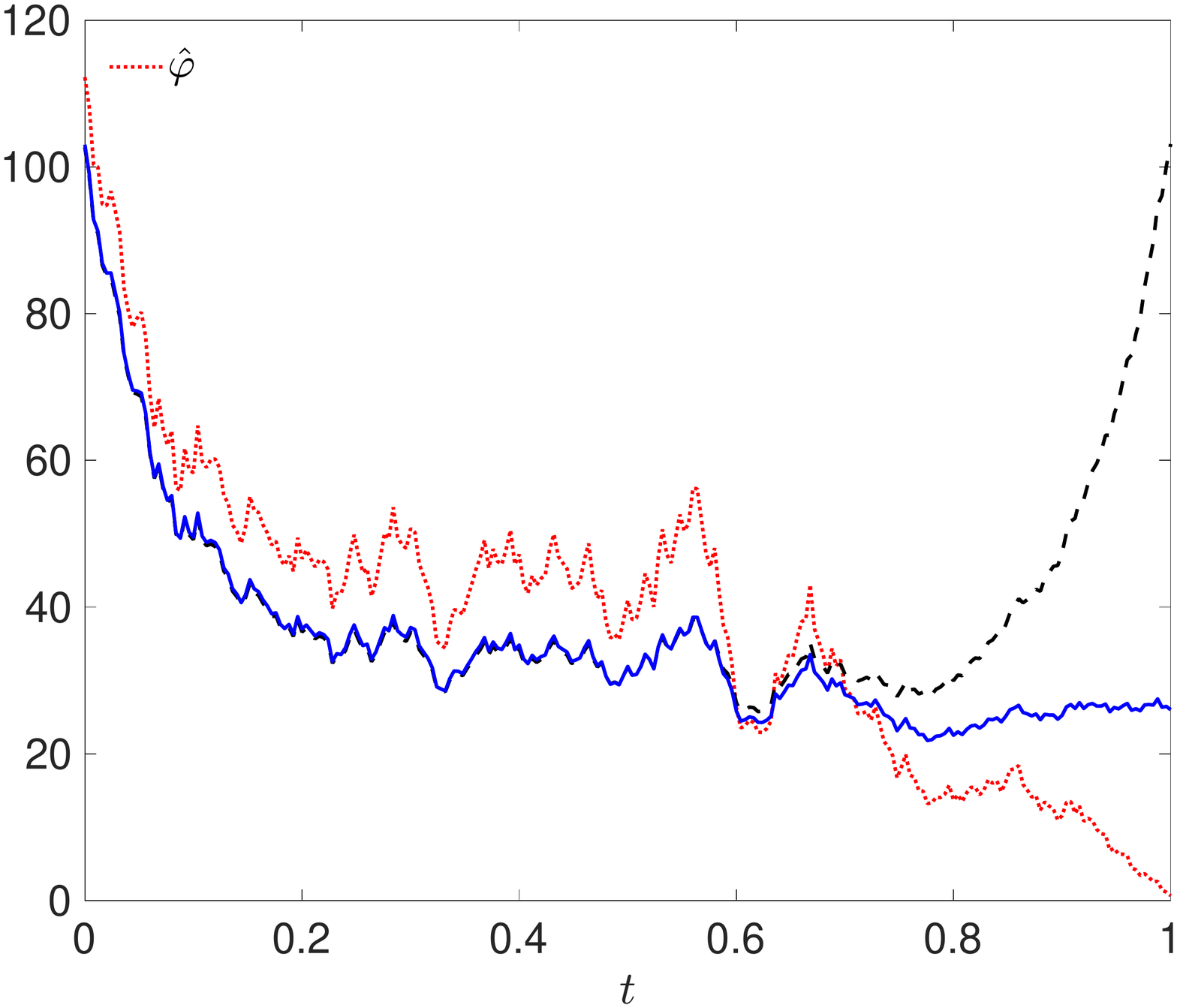} & \includegraphics[width=0.33\textwidth]{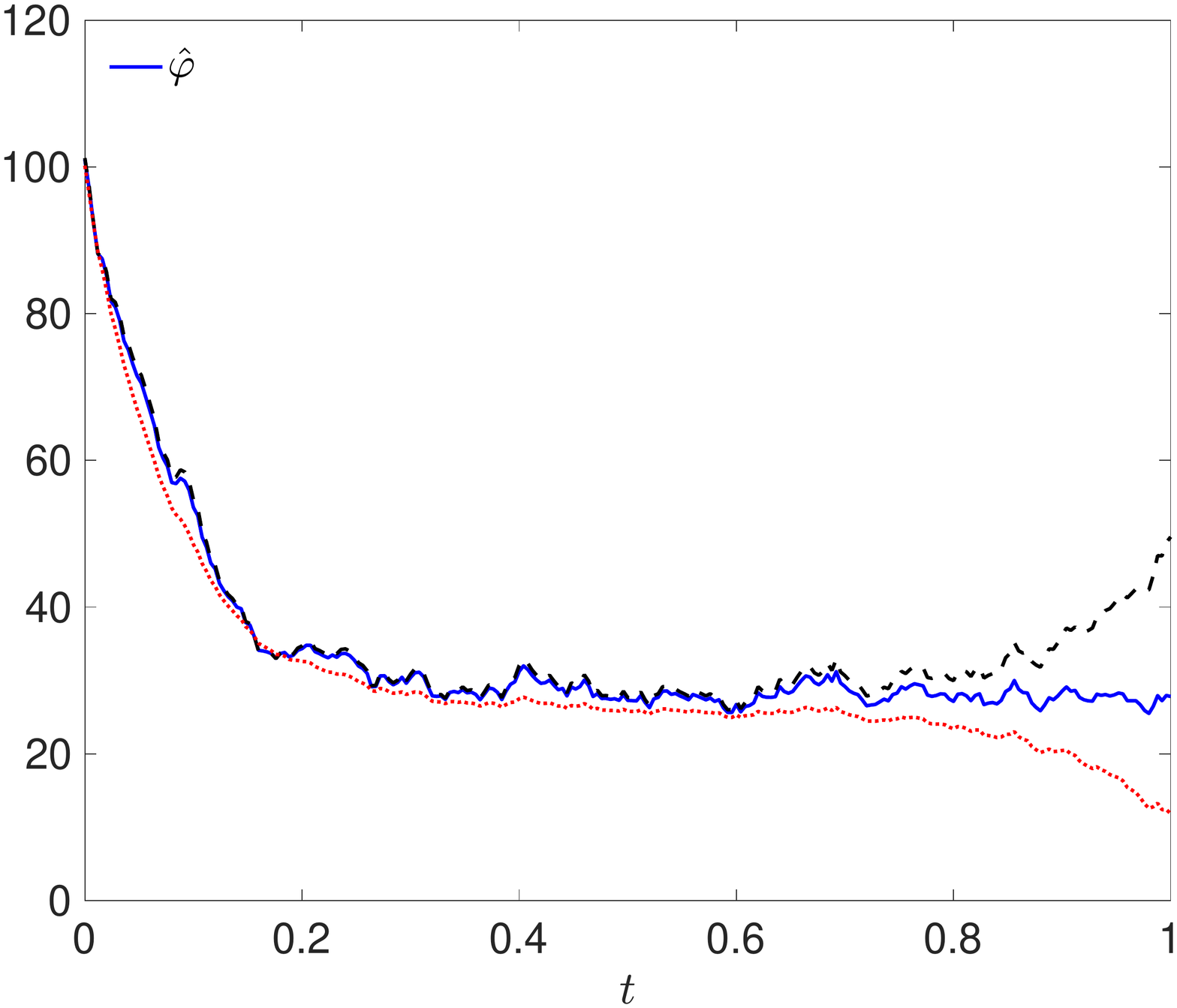} \\
\hspace{-5mm}\includegraphics[width=0.33\textwidth]{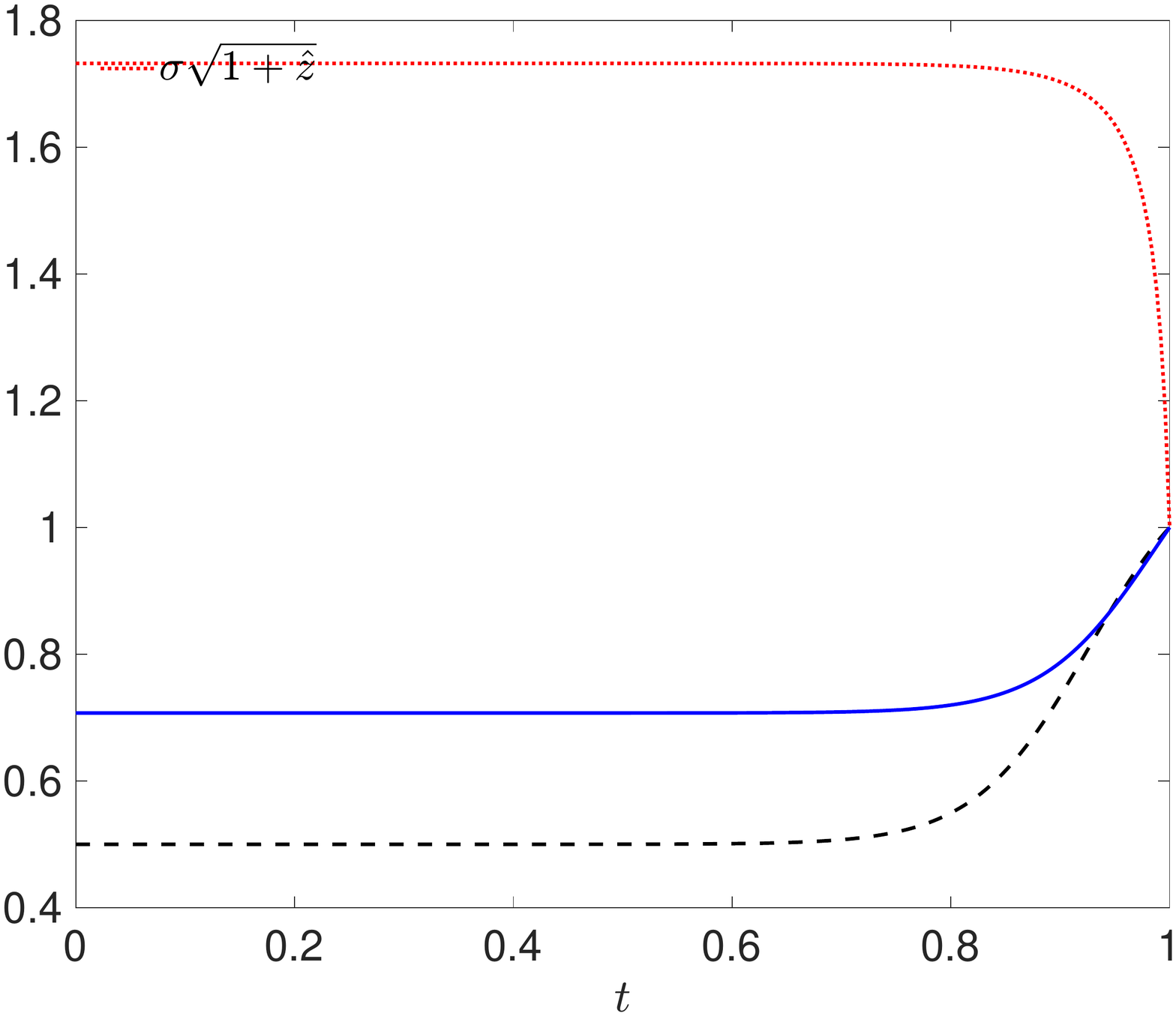} & \includegraphics[width=0.33\textwidth]{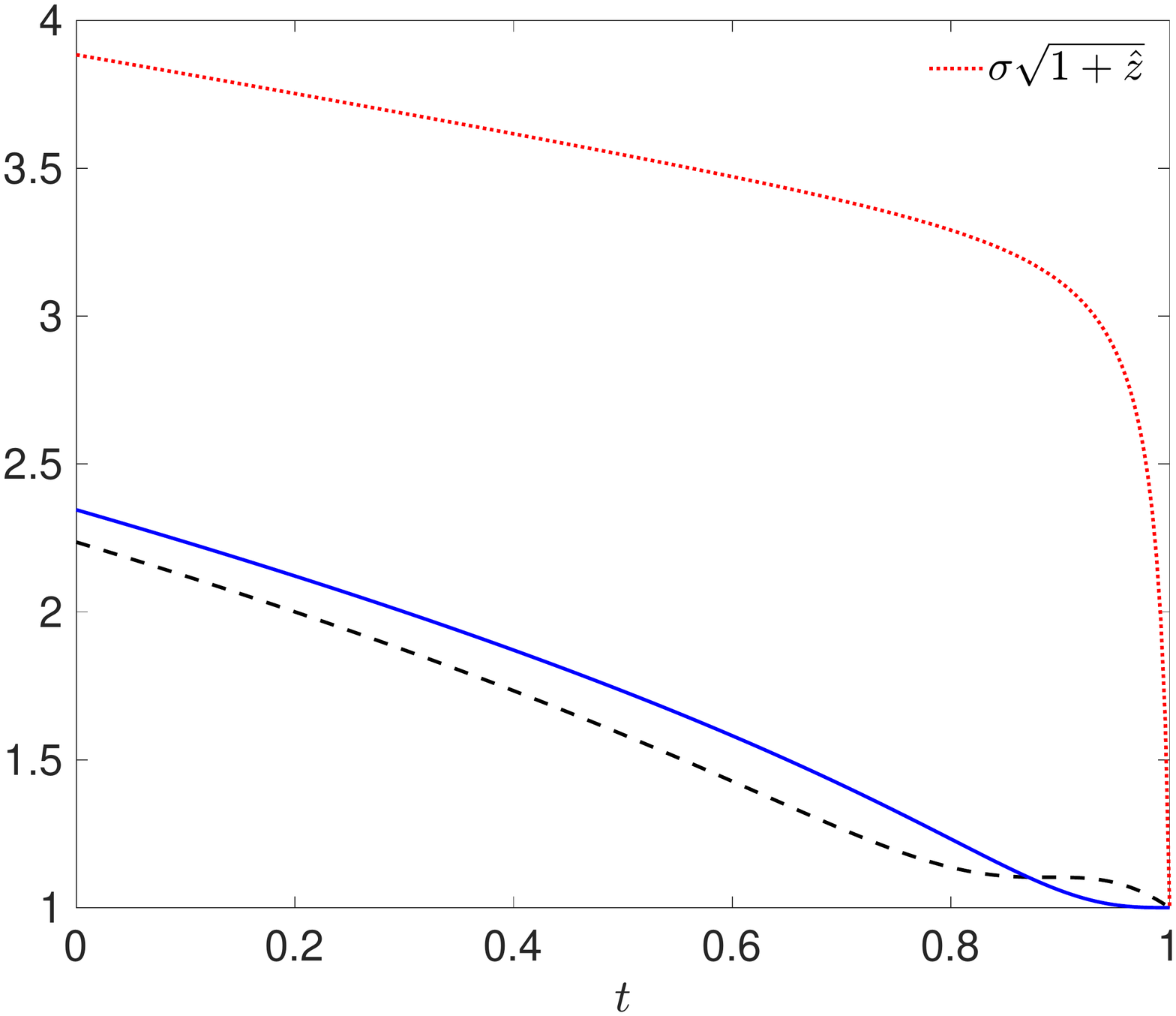} & \includegraphics[width=0.33\textwidth]{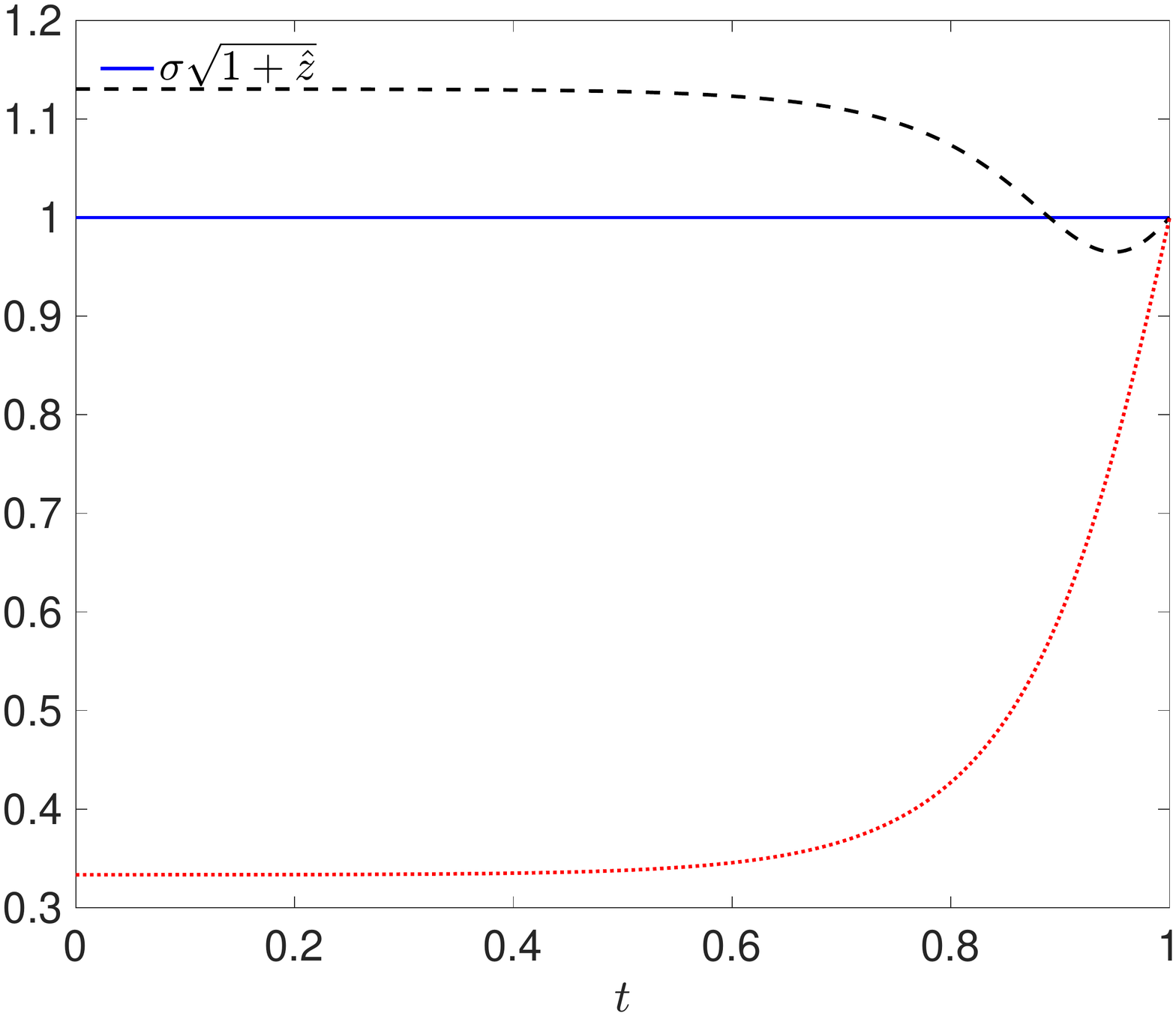} \\
\hspace{-5mm}\includegraphics[width=0.33\textwidth]{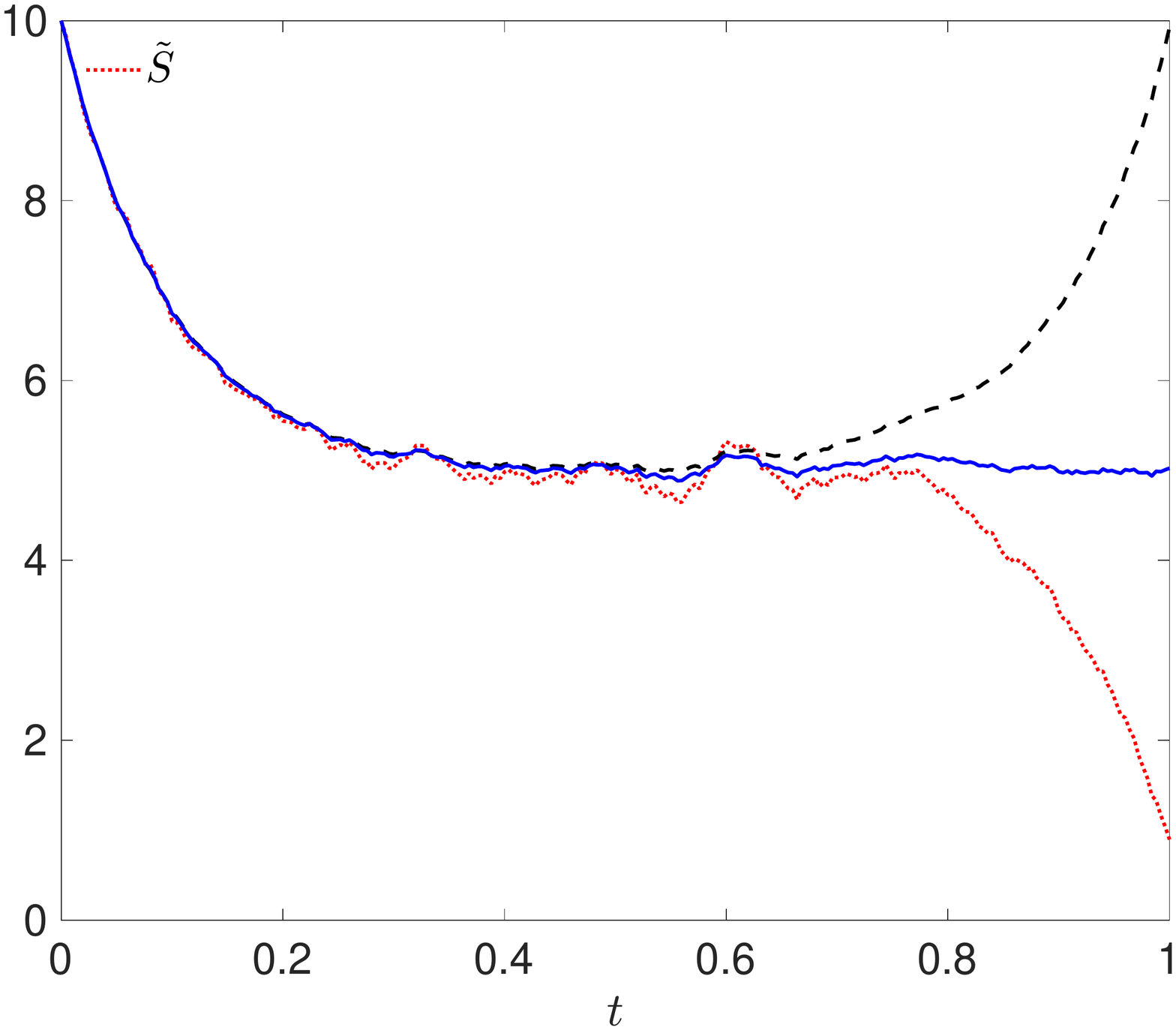}& \includegraphics[width=0.33\textwidth]{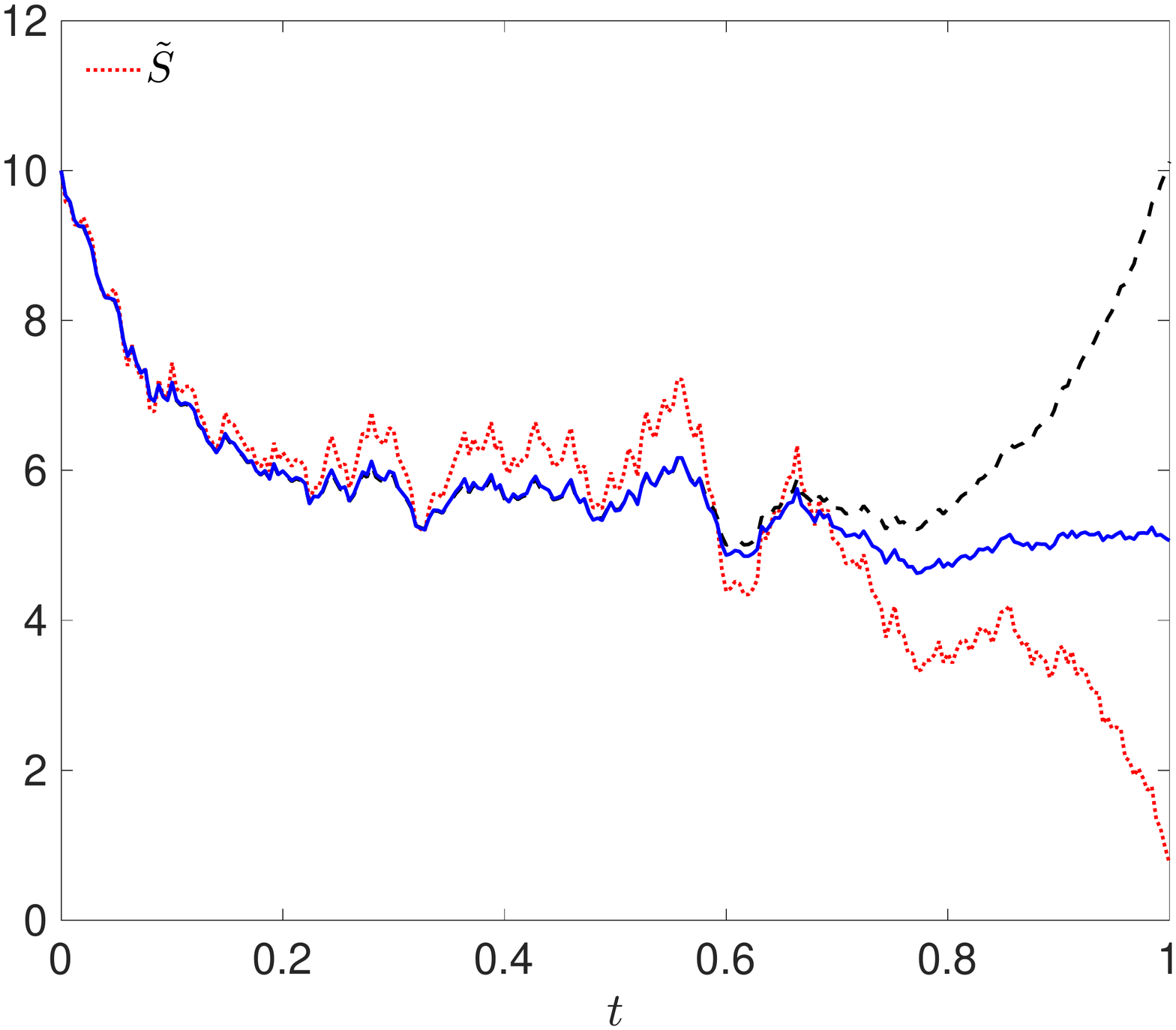} & \includegraphics[width=0.33\textwidth]{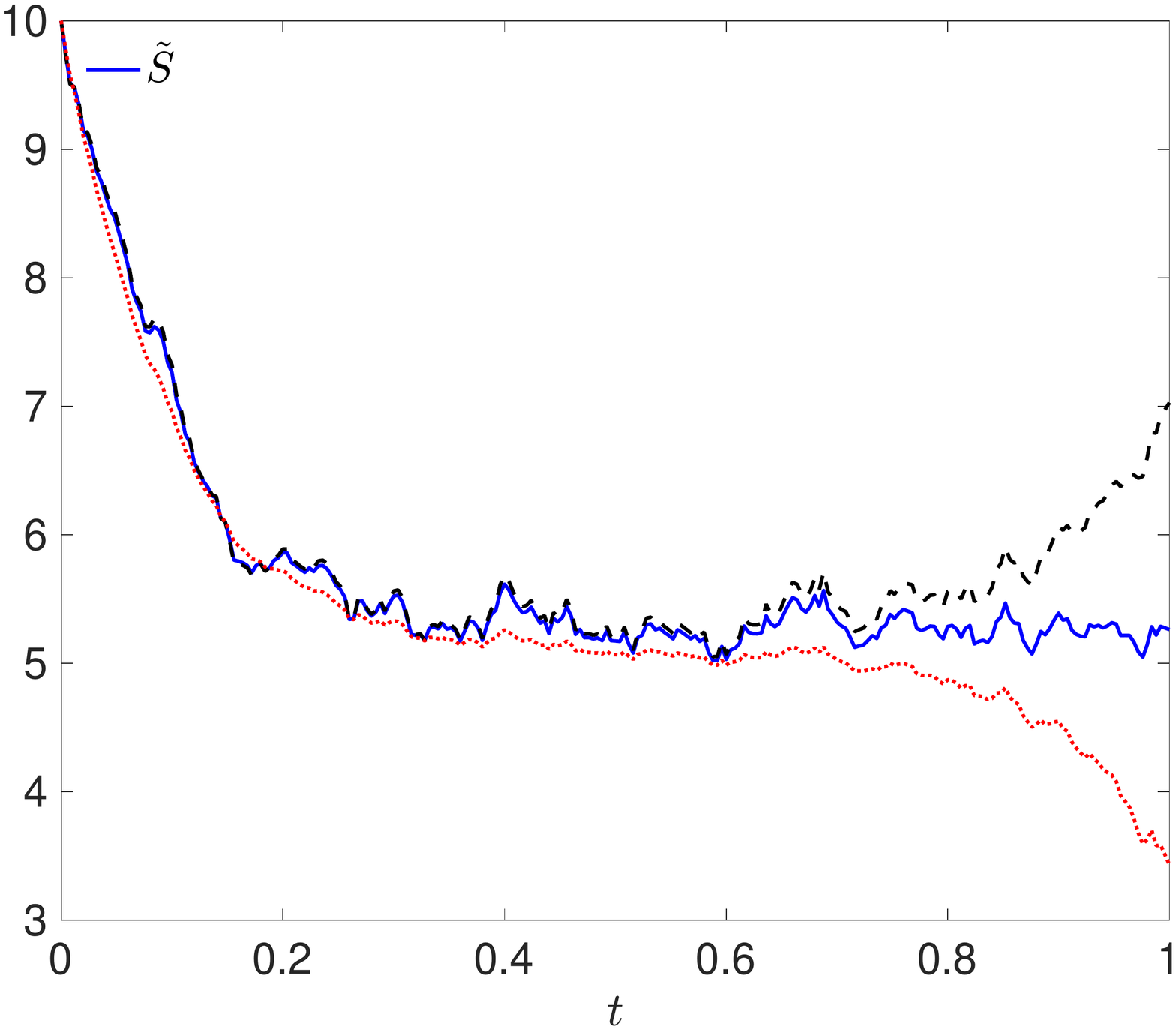}
\end{tabular}
\caption{{\small Optimal production rate $\widehat q$, derivative price $\widehat \varphi$, volatility $\widehat \sigma$ and commodity price $\widehat S$ when the producer has no derivative position $\lambda =0$ (blue), bought the derivative $\lambda < 0$ (black) and sold the derivative, $\lambda > 0$ (red). Parameter values: $s_0 = 10$,  $a = 0.5$, $g = 0.1$,  $\kappa = 0.01$,  $\sigma = 1$, $T = 1$, $\mu = 0.0$, $q_0=0$, $\lambda \in \{ -0.1, 1\}$ for Model 1 \& 2, $\lambda \in \{ -0.05, 0.1\}$ for Model 3.}}
\label{fig:sim}
\end{flushleft}
\end{figure}

\begin{figure}[hbt!]
\begin{flushleft}
\begin{tabular}{c c c} 
Model 1 & Model 2 & Model 3 \\
\hspace{-5mm}\includegraphics[width=0.33\textwidth]{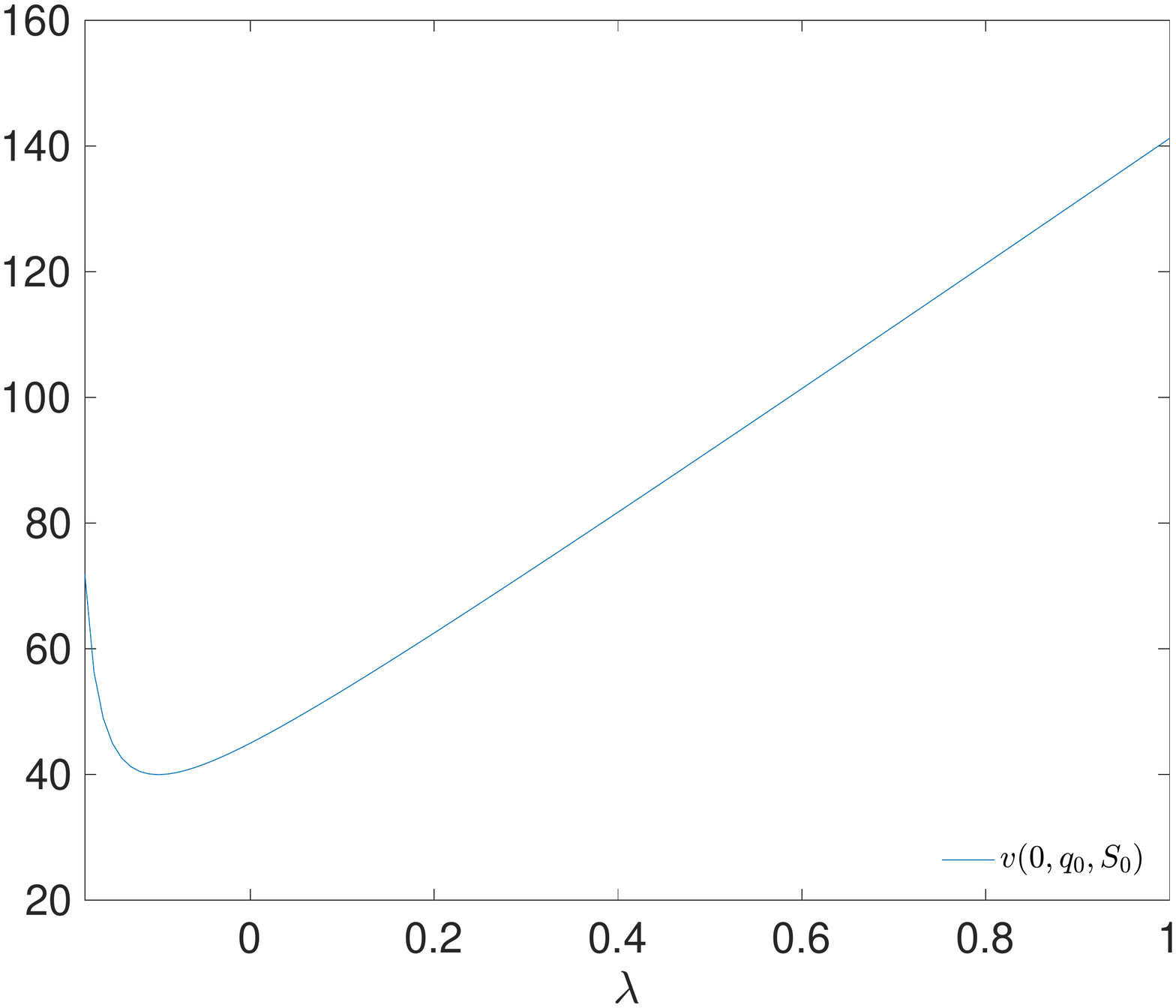} & \includegraphics[width=0.33\textwidth]{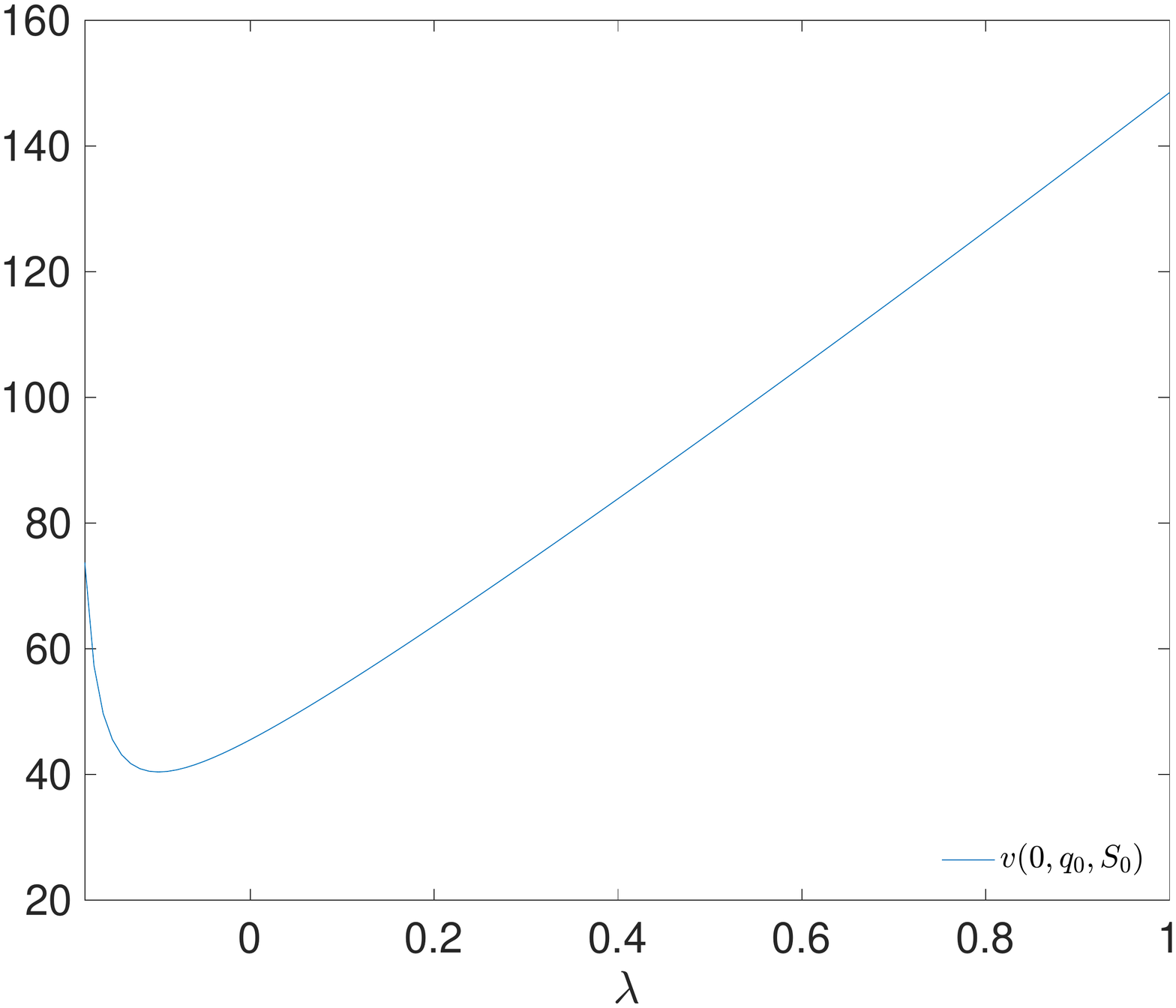} & \includegraphics[width=0.33\textwidth]{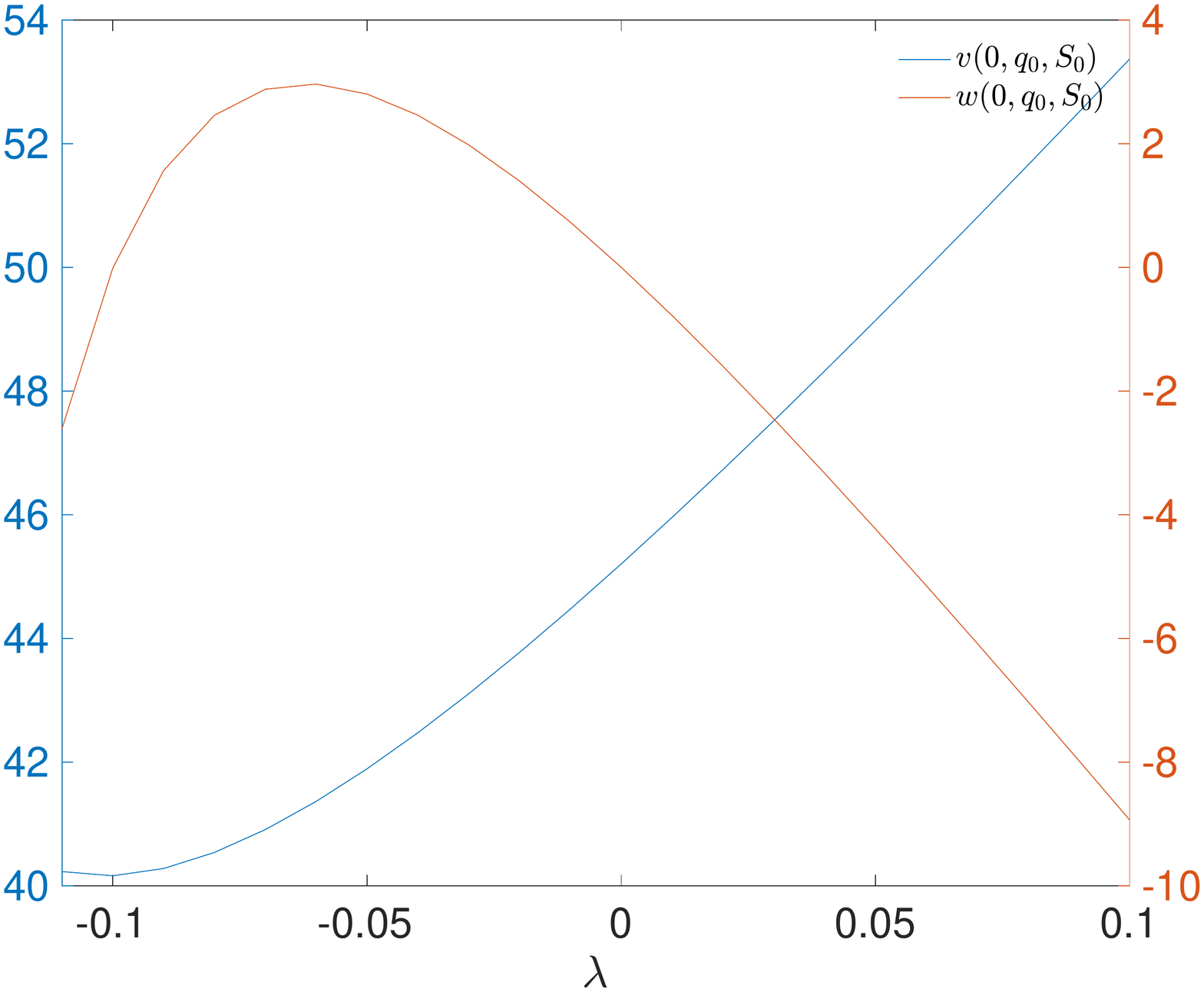}
\end{tabular}
\caption{{\small Producer's and trader's value function at initial time as a function of derivative position $\lambda$. Parameter values: $s_0 = 10$,  $a = 0.5$, $g = 0.1$,  $\kappa = 0.01$,  $\sigma = 1$, $T = 1$, $\mu = 0.0$, $q_0=0$.}}
\label{fig:v0w0}
\end{flushleft}
\end{figure}

\begin{figure}[t!]
\begin{center}
\begin{tabular}{c c c} 
Model 1 & Model 2 & Model 3 \\
\hspace{-5mm}\includegraphics[width=0.33\textwidth]{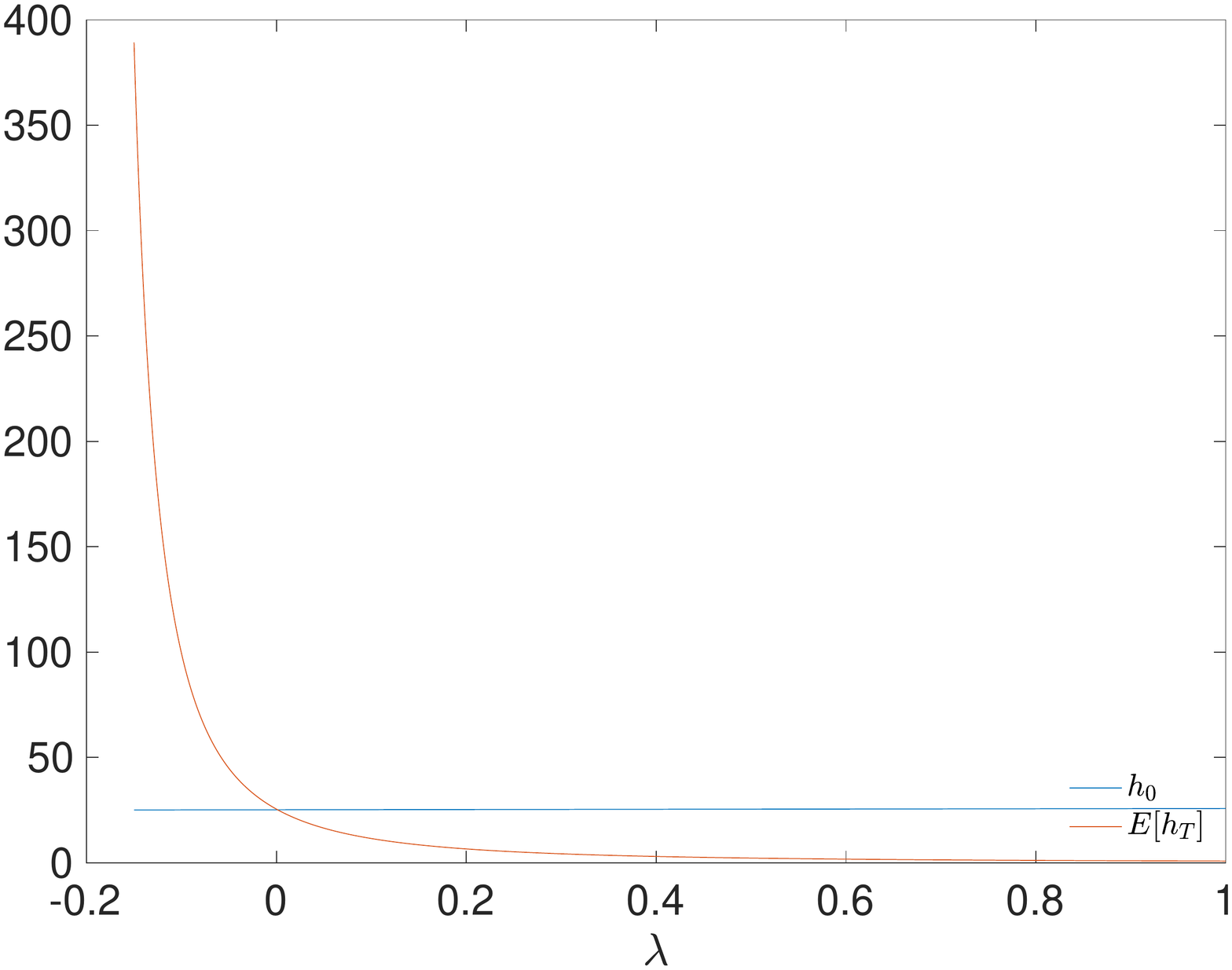} & \includegraphics[width=0.33\textwidth]{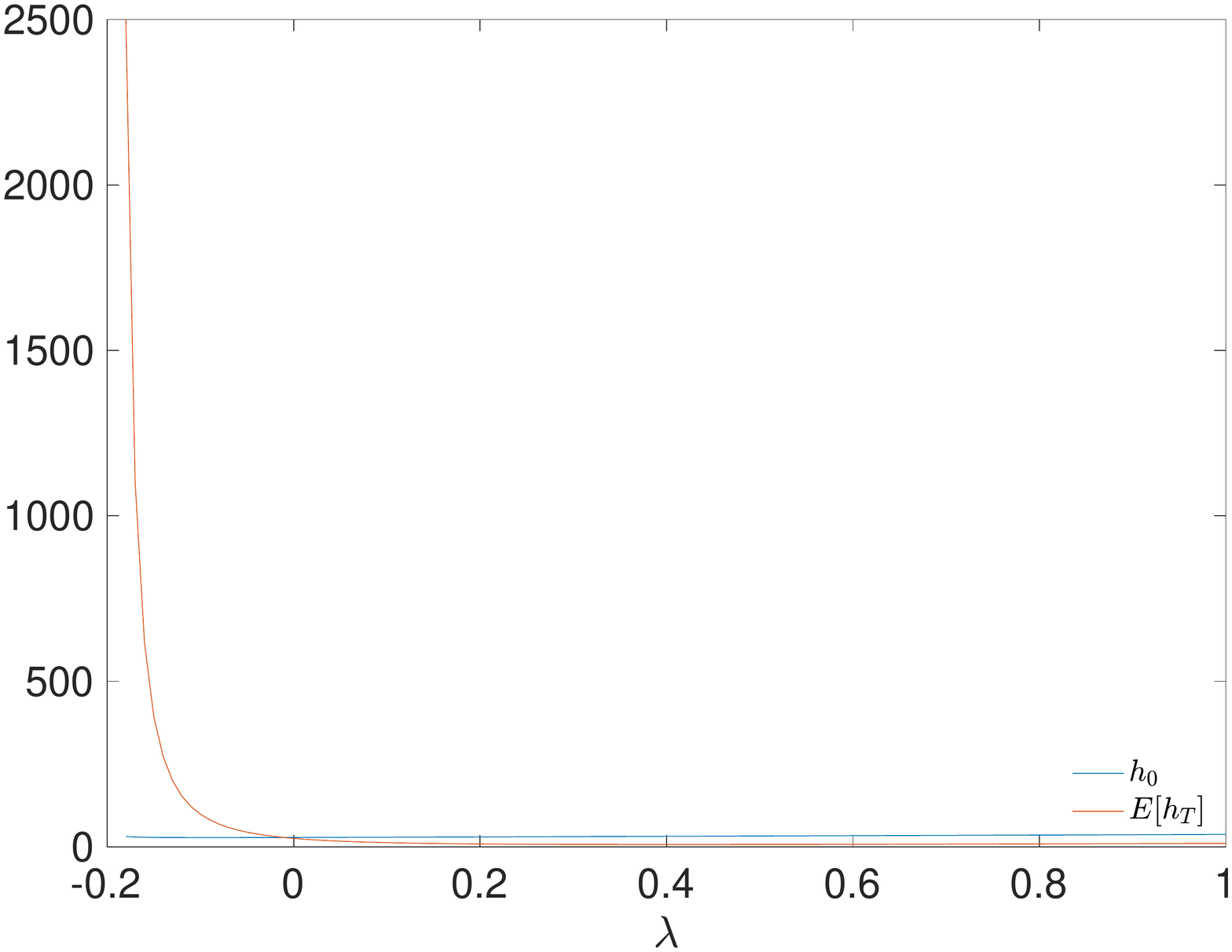} & \includegraphics[width=0.33\textwidth]{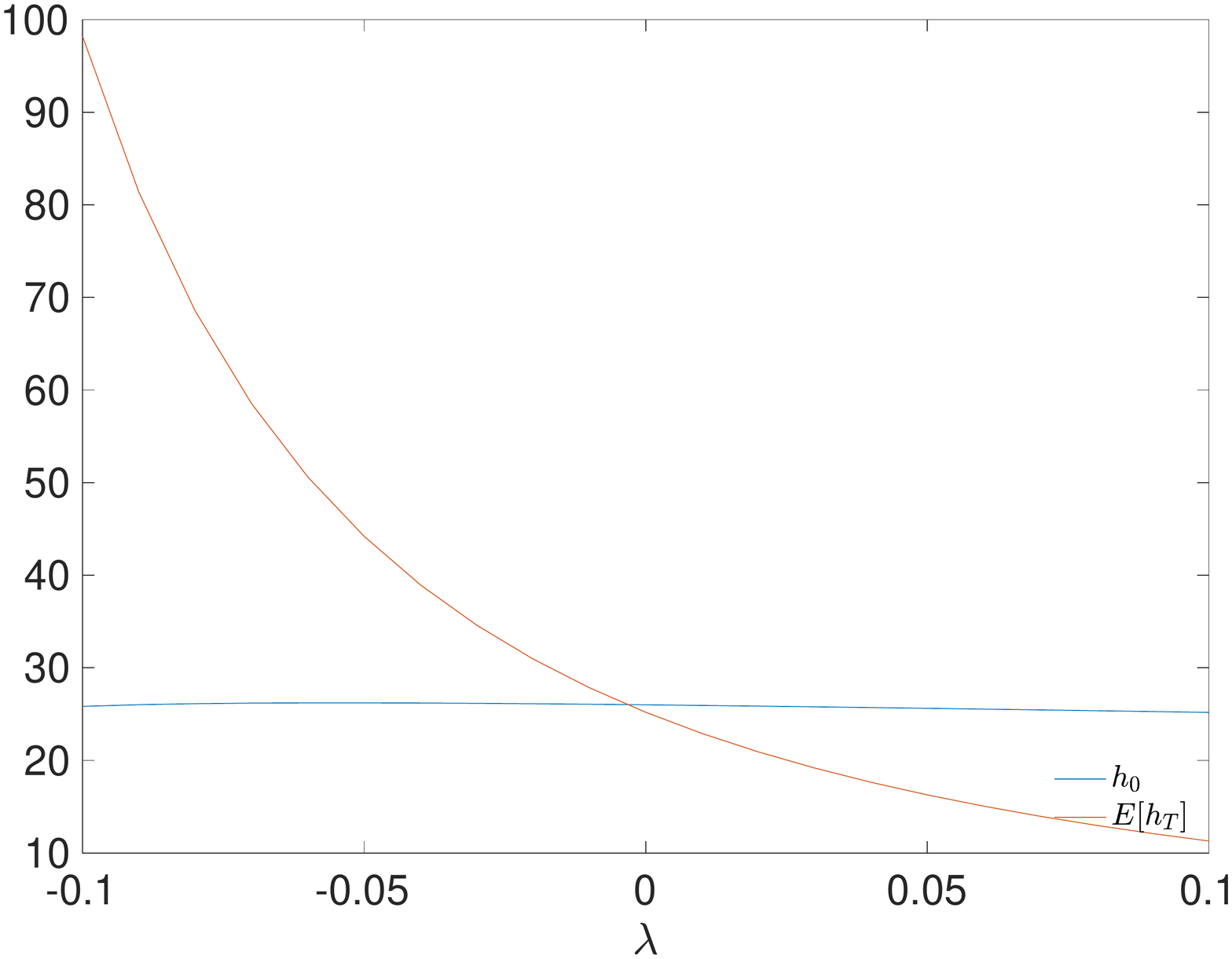} \\ 
\hspace{-5mm}\includegraphics[width=0.33\textwidth]{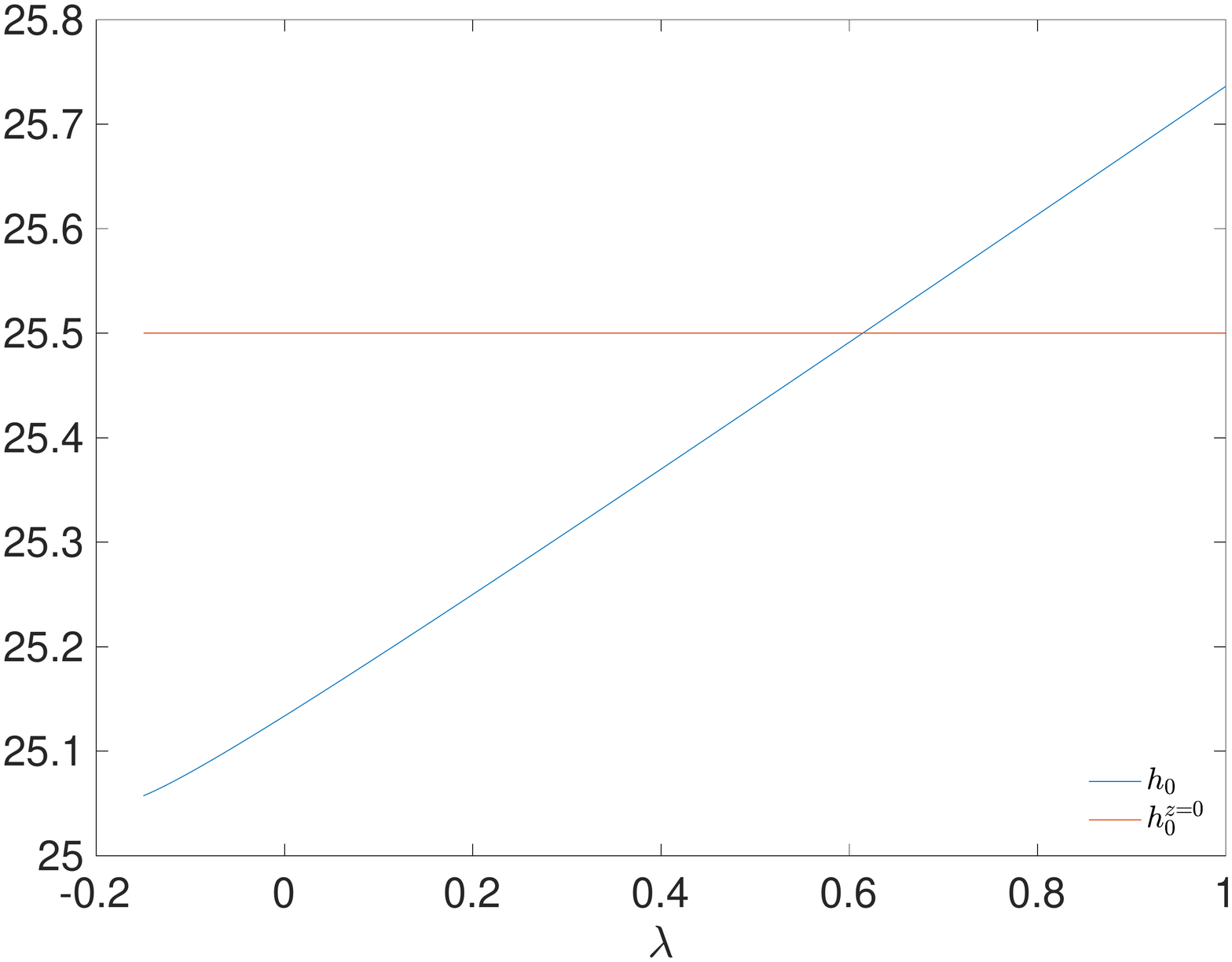} & \includegraphics[width=0.33\textwidth]{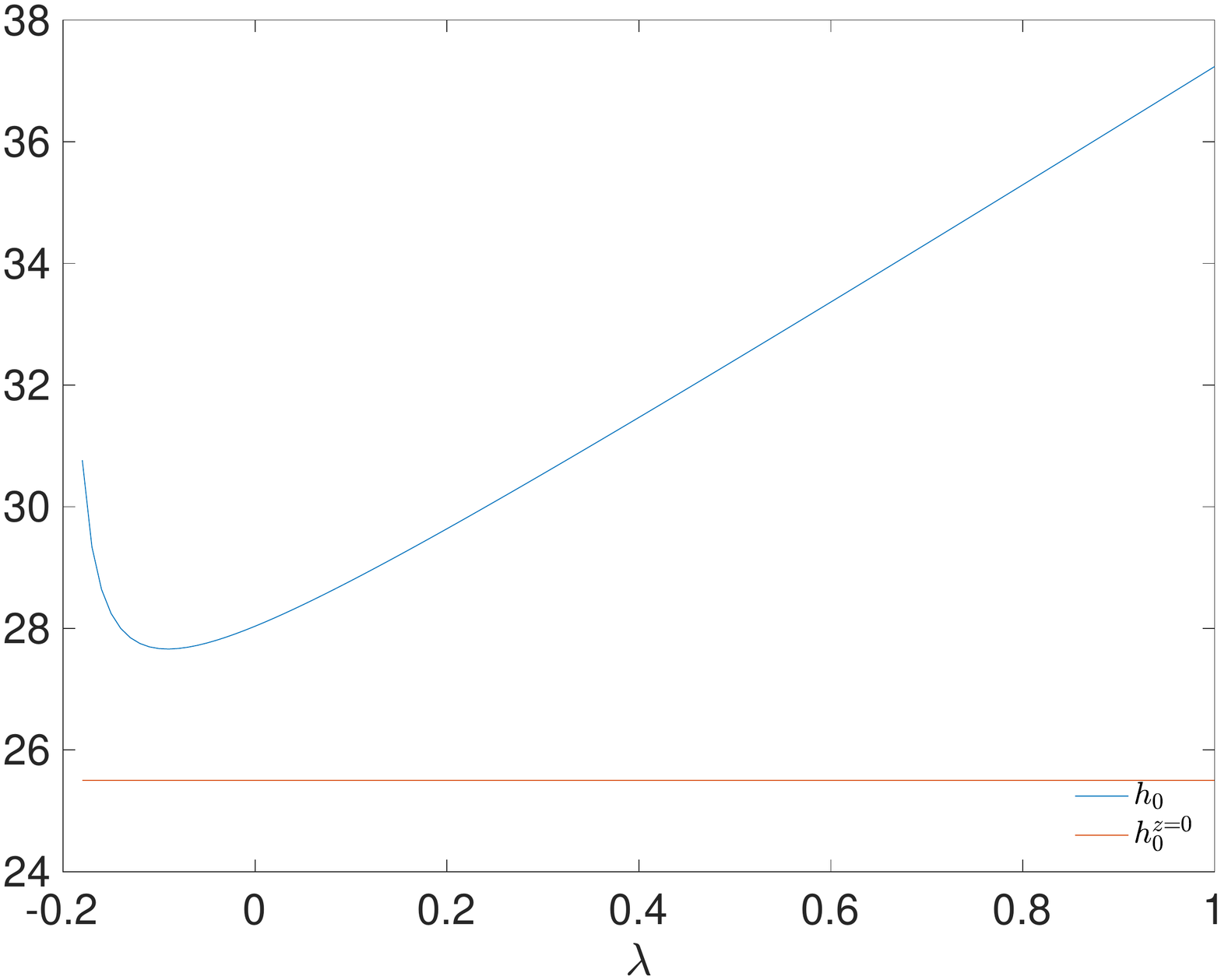} & \includegraphics[width=0.33\textwidth]{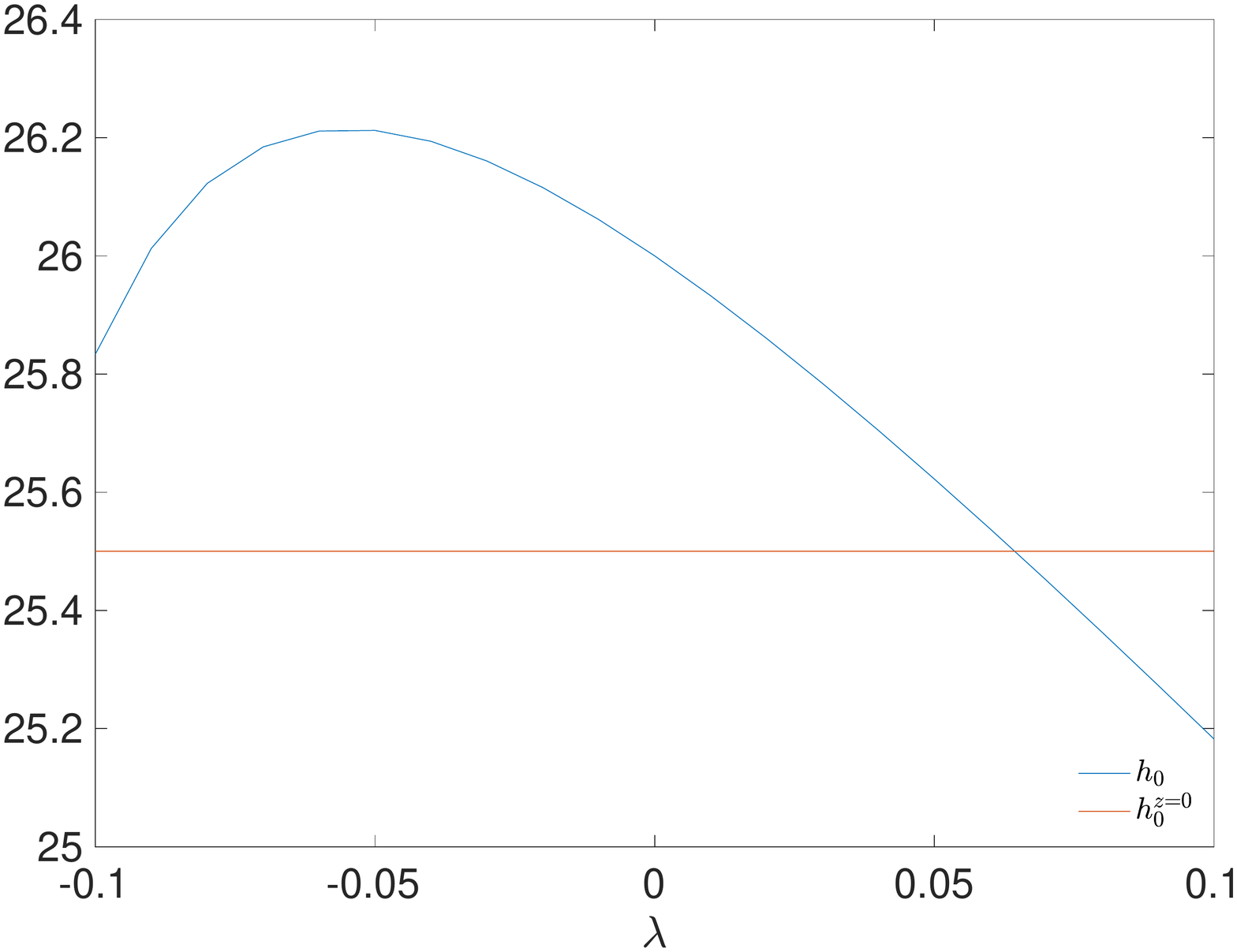} \\
\end{tabular}
\caption{{\small Values of the derivative $h_0$, the expected value of the payoff $\E^\P[h_T]$ and the value of the derivative in case of no volatility manipulation $h_0^{z=0}$ as a function of the net position $\lambda$. Parameter values: $s_0 = 10$,  $a = 0.5$, $g = 0.1$,  $\kappa = 0.01$,  $\sigma = 1$, $T = 1$, $\mu = 0.0$, $q_0=q^\star$.}}
\label{fig:h0}
\end{center}
\end{figure}

\begin{figure}[hbt!]
\vspace{-9mm}
\begin{center}
\begin{tabular}{c c c} 
$a =  0.1$ & $a=0.5$ & $a = 0.9$ \\
\hspace{-5mm}\includegraphics[width=0.33\textwidth]{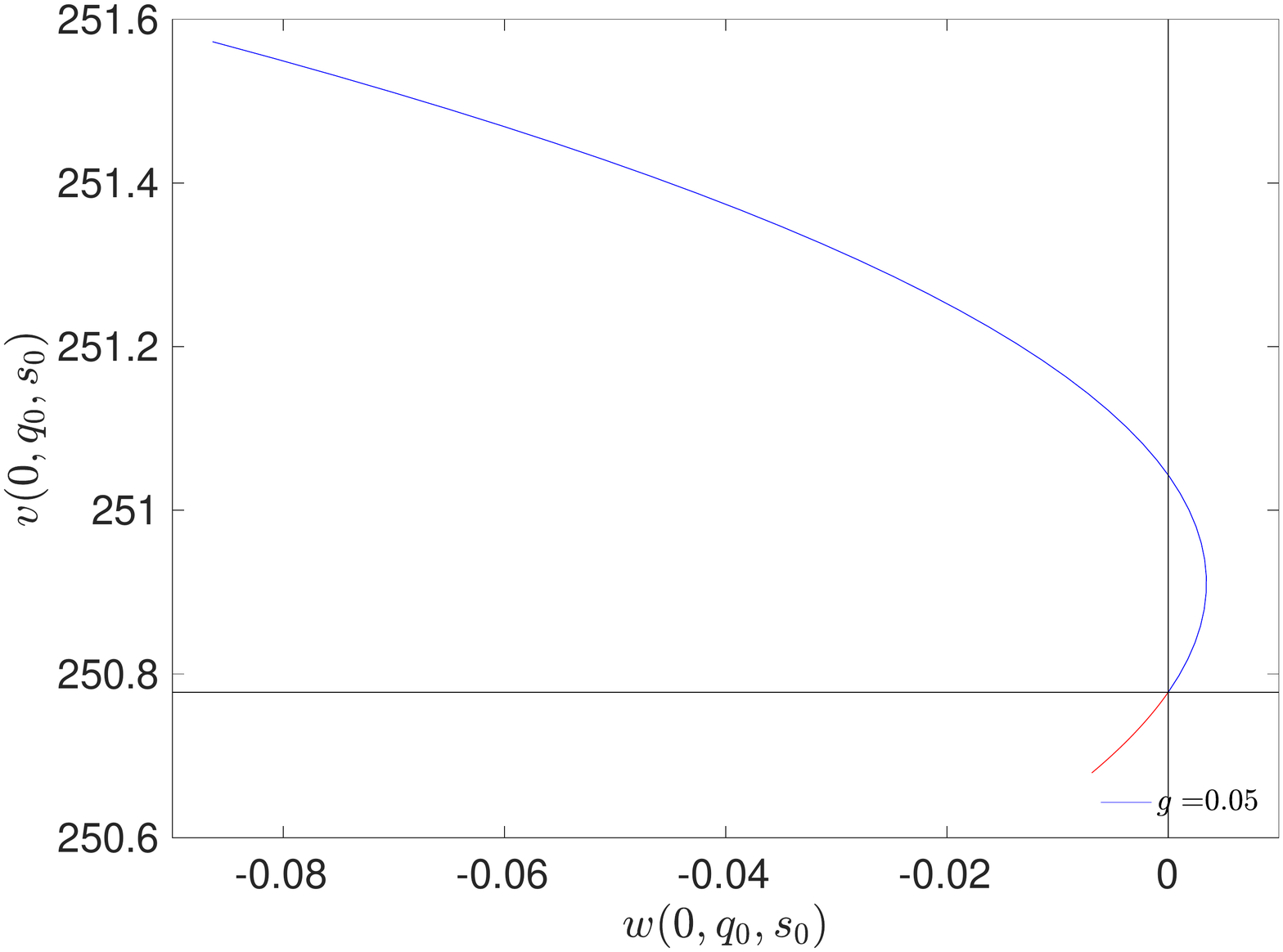} & \includegraphics[width=0.33\textwidth]{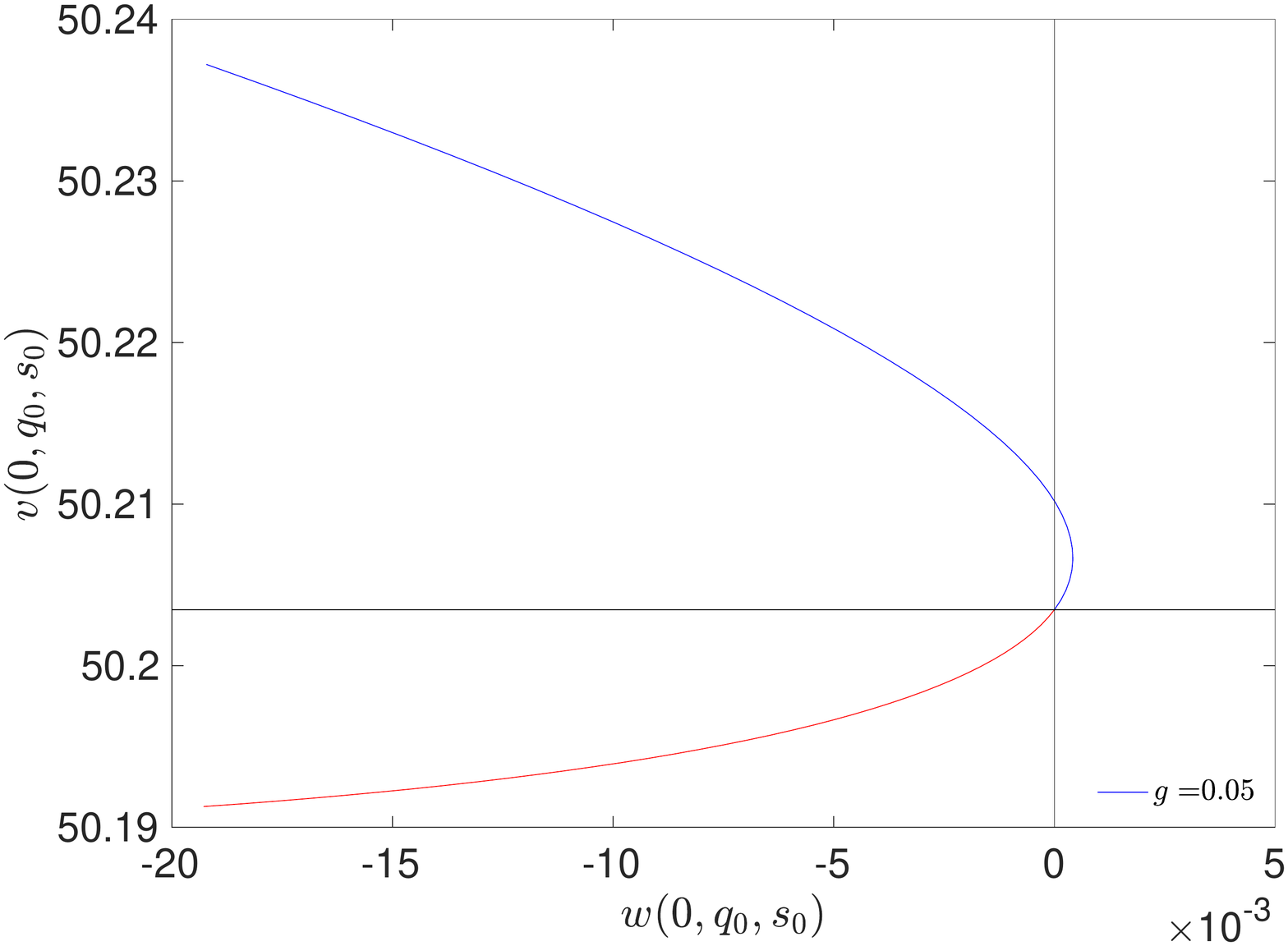} & \includegraphics[width=0.33\textwidth]{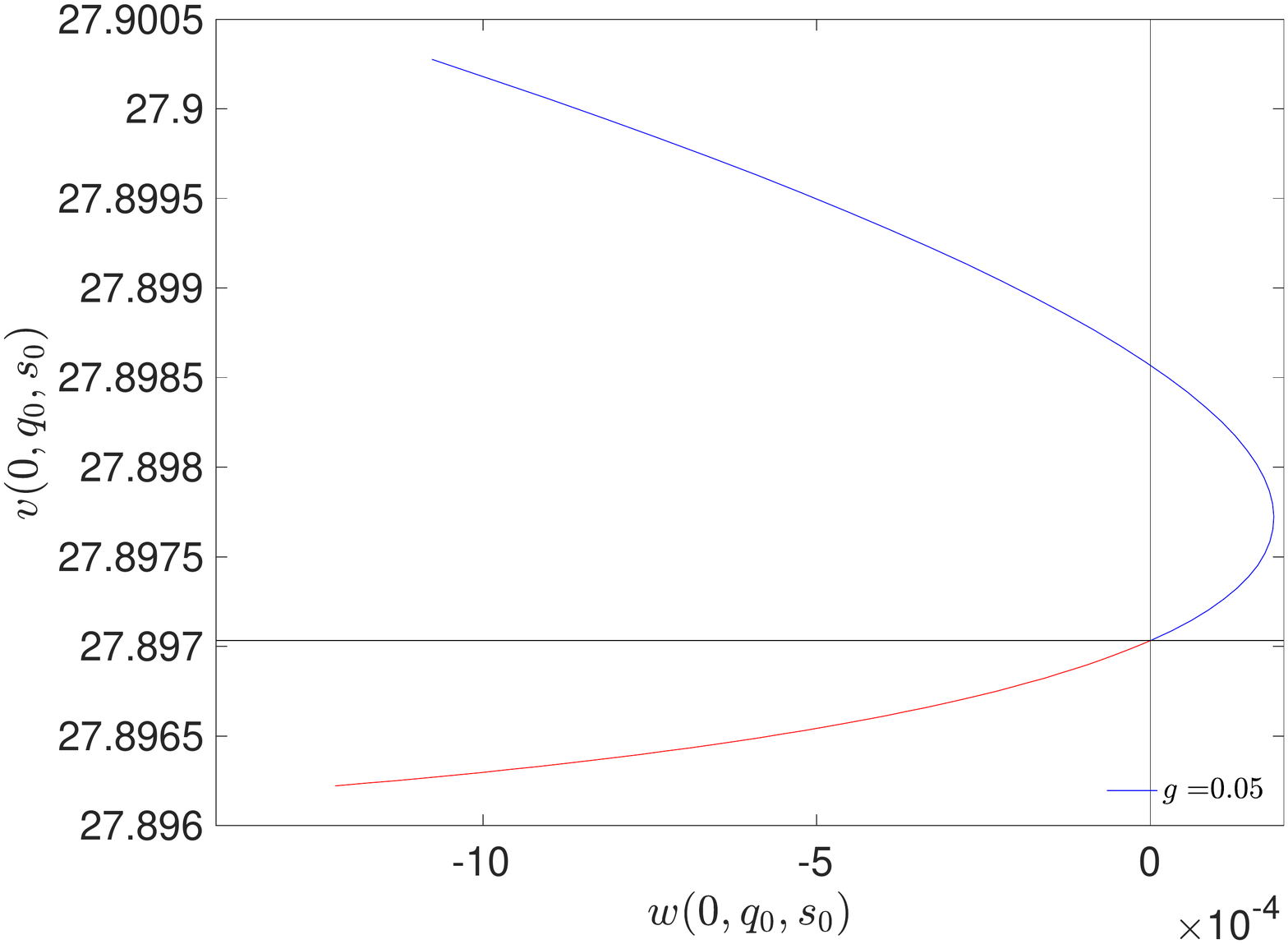} \\
\hspace{-5mm}\includegraphics[width=0.33\textwidth]{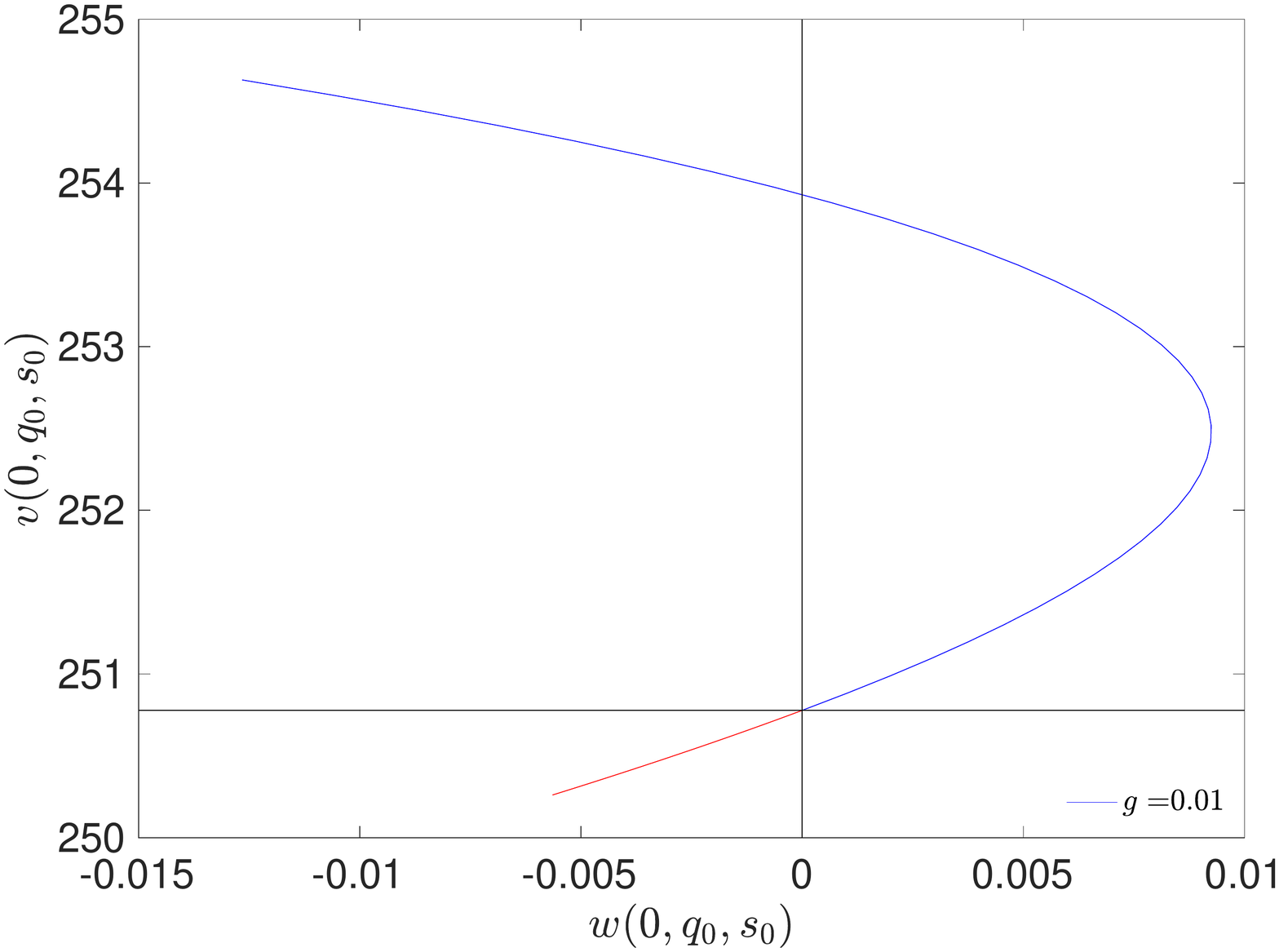}  & \includegraphics[width=0.33\textwidth]{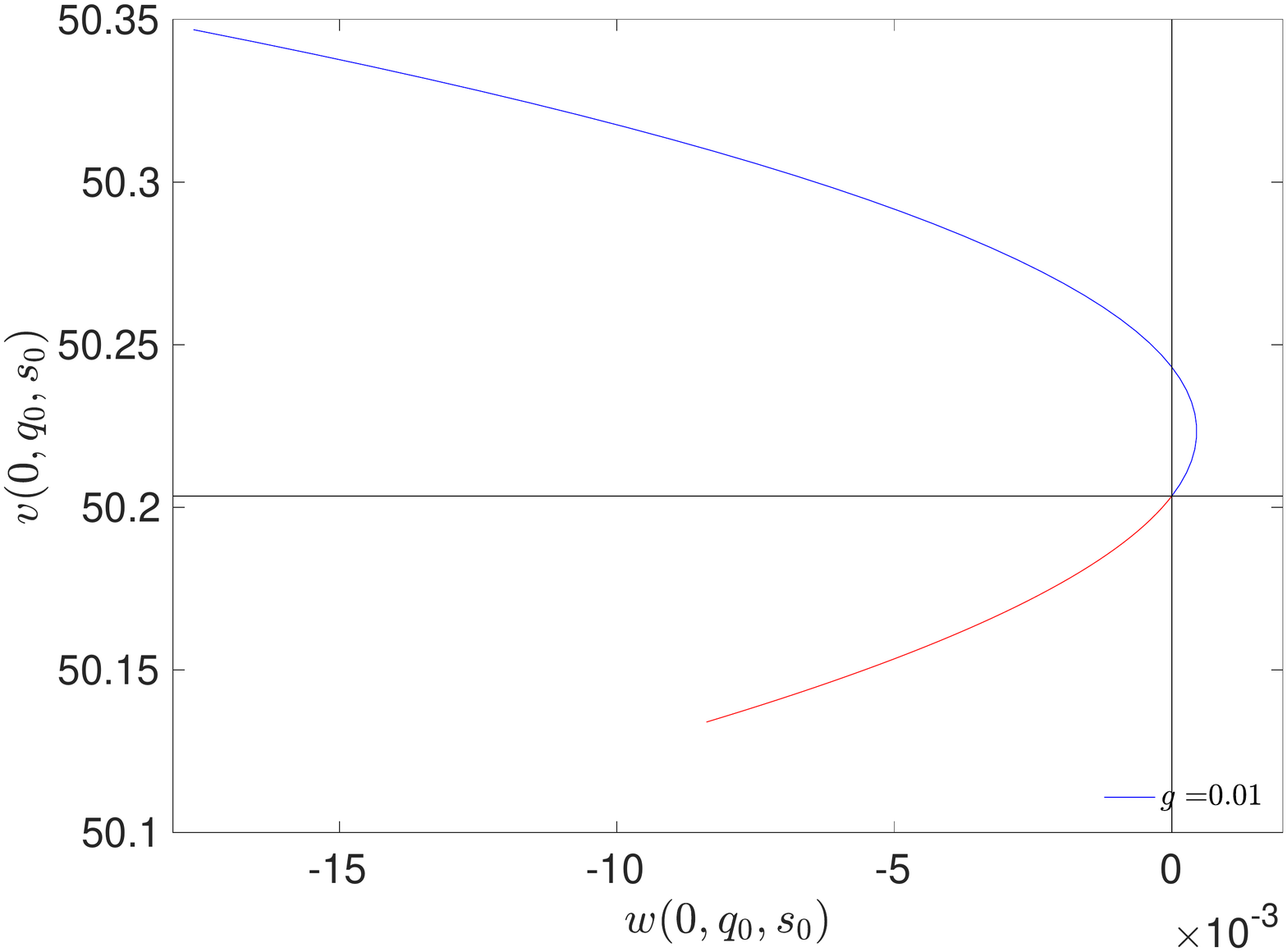} & \includegraphics[width=0.33\textwidth]{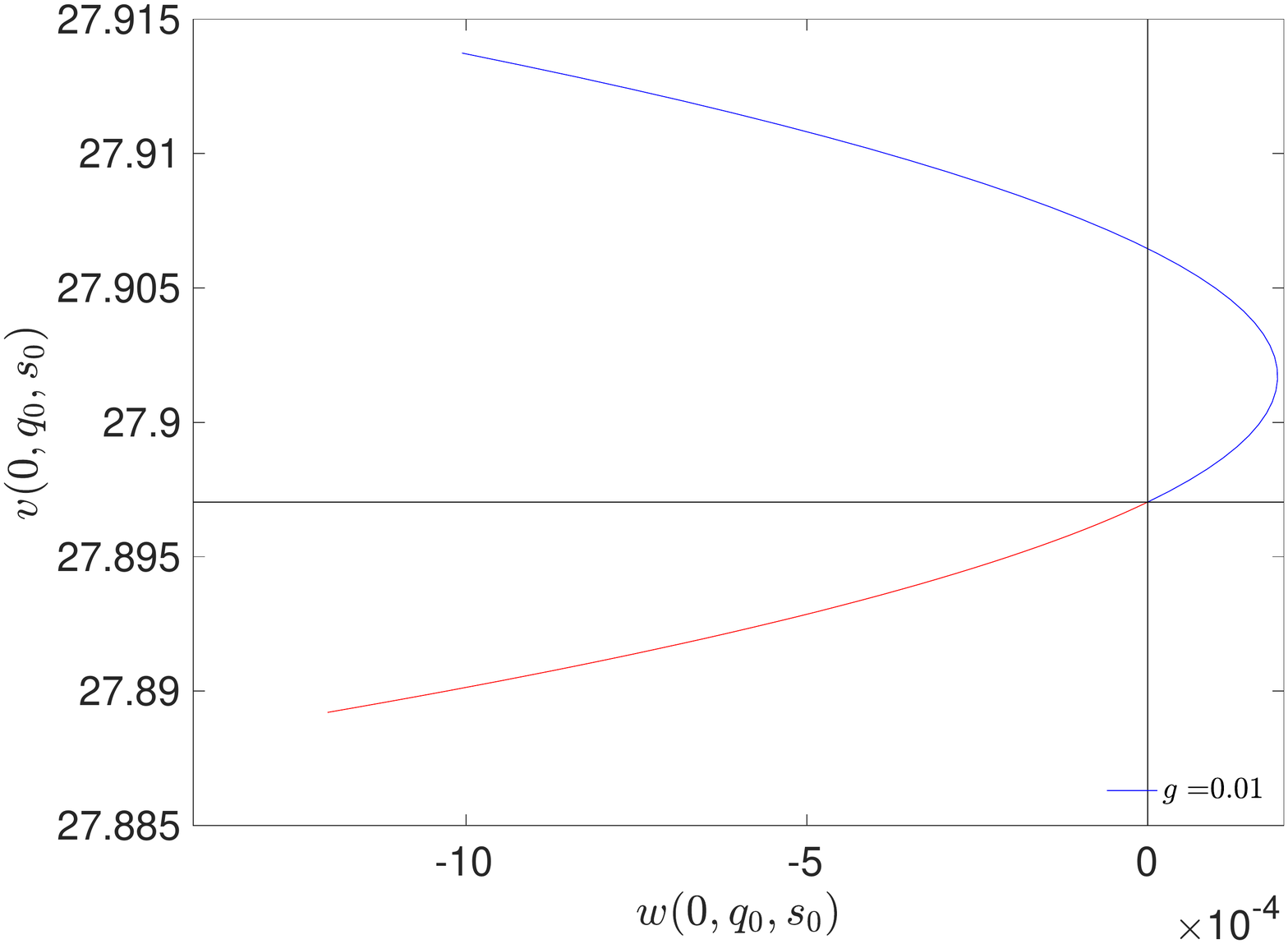}  \\
\hspace{-5mm}\includegraphics[width=0.33\textwidth]{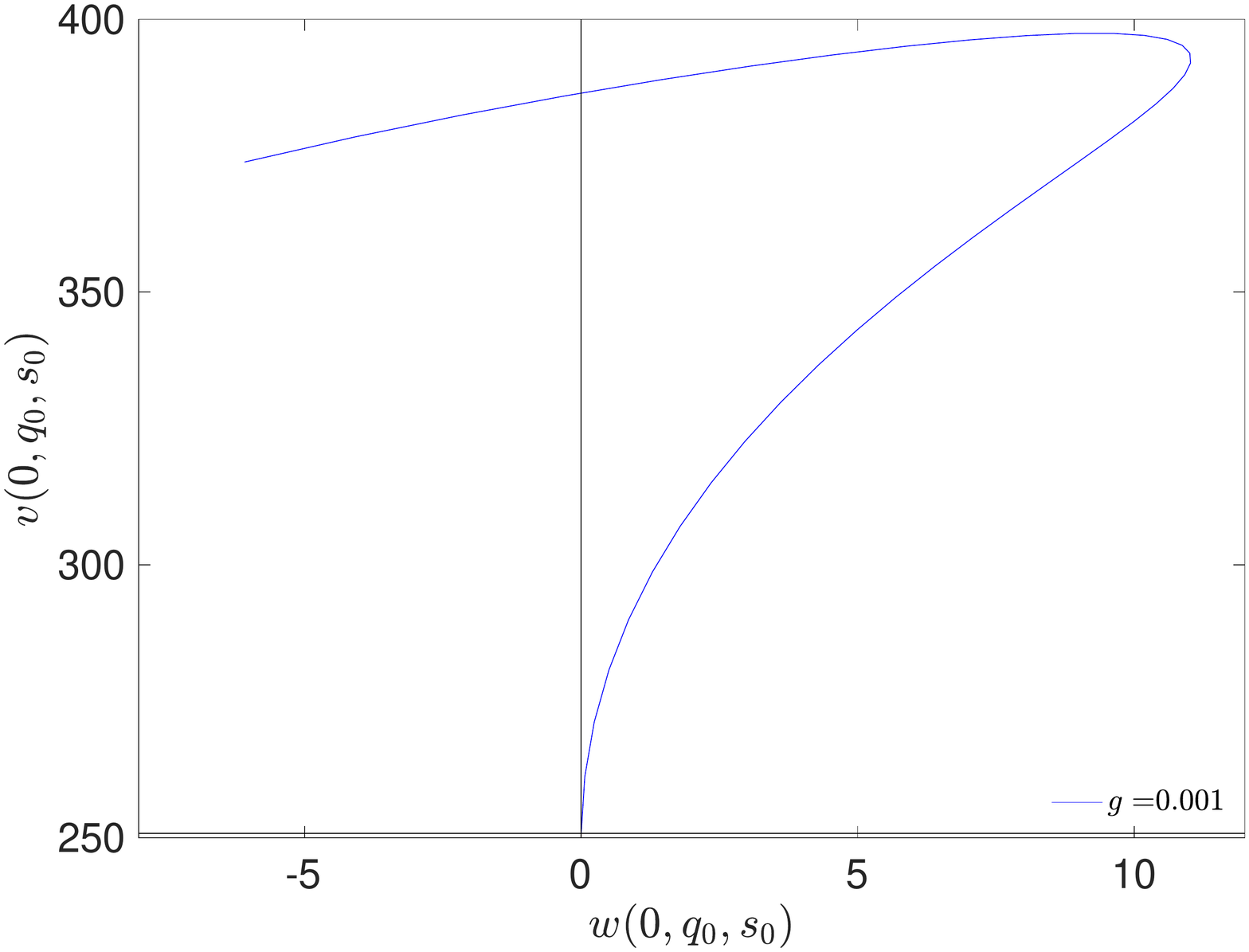} & \includegraphics[width=0.33\textwidth]{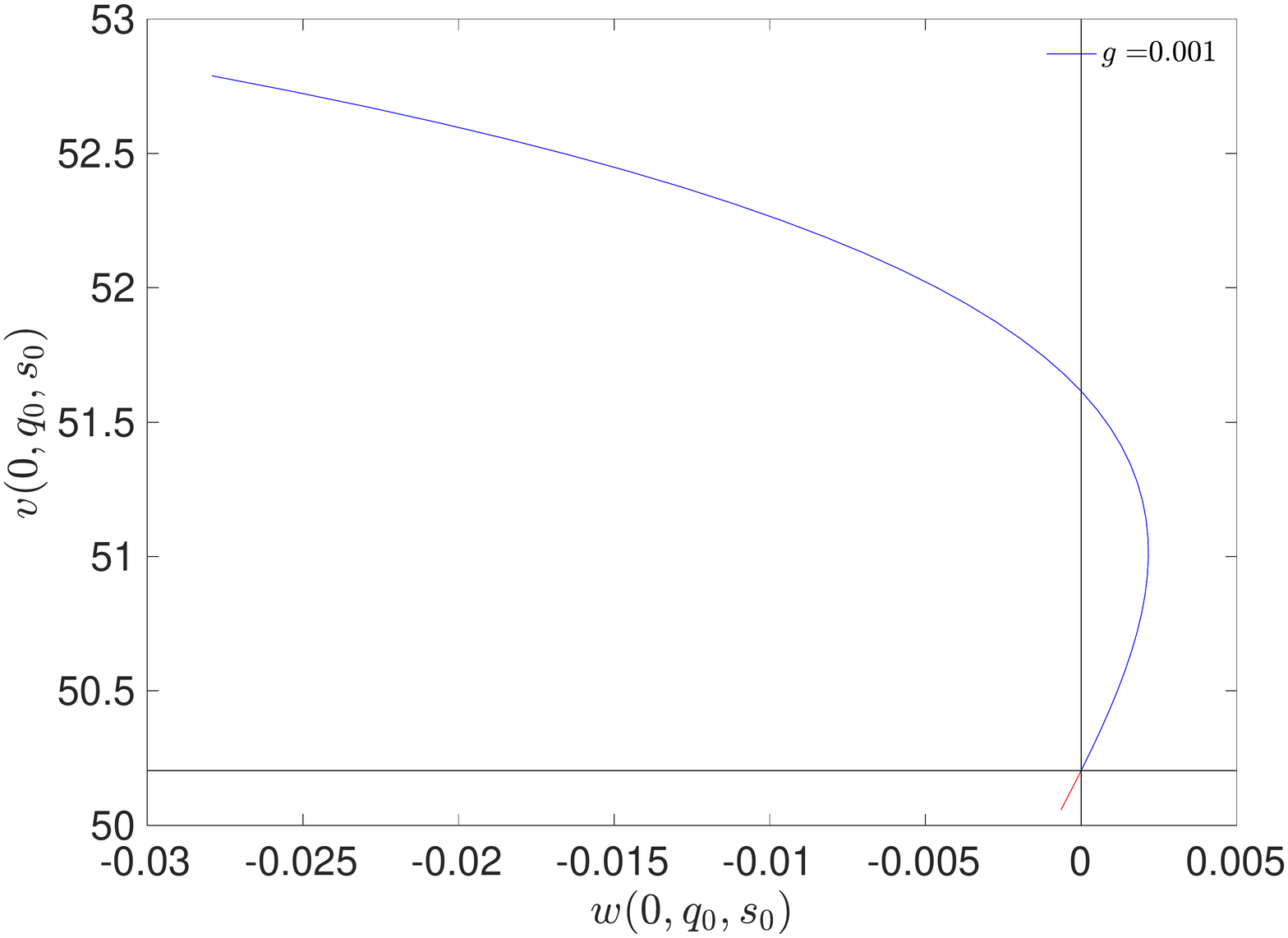} & \includegraphics[width=0.33\textwidth]{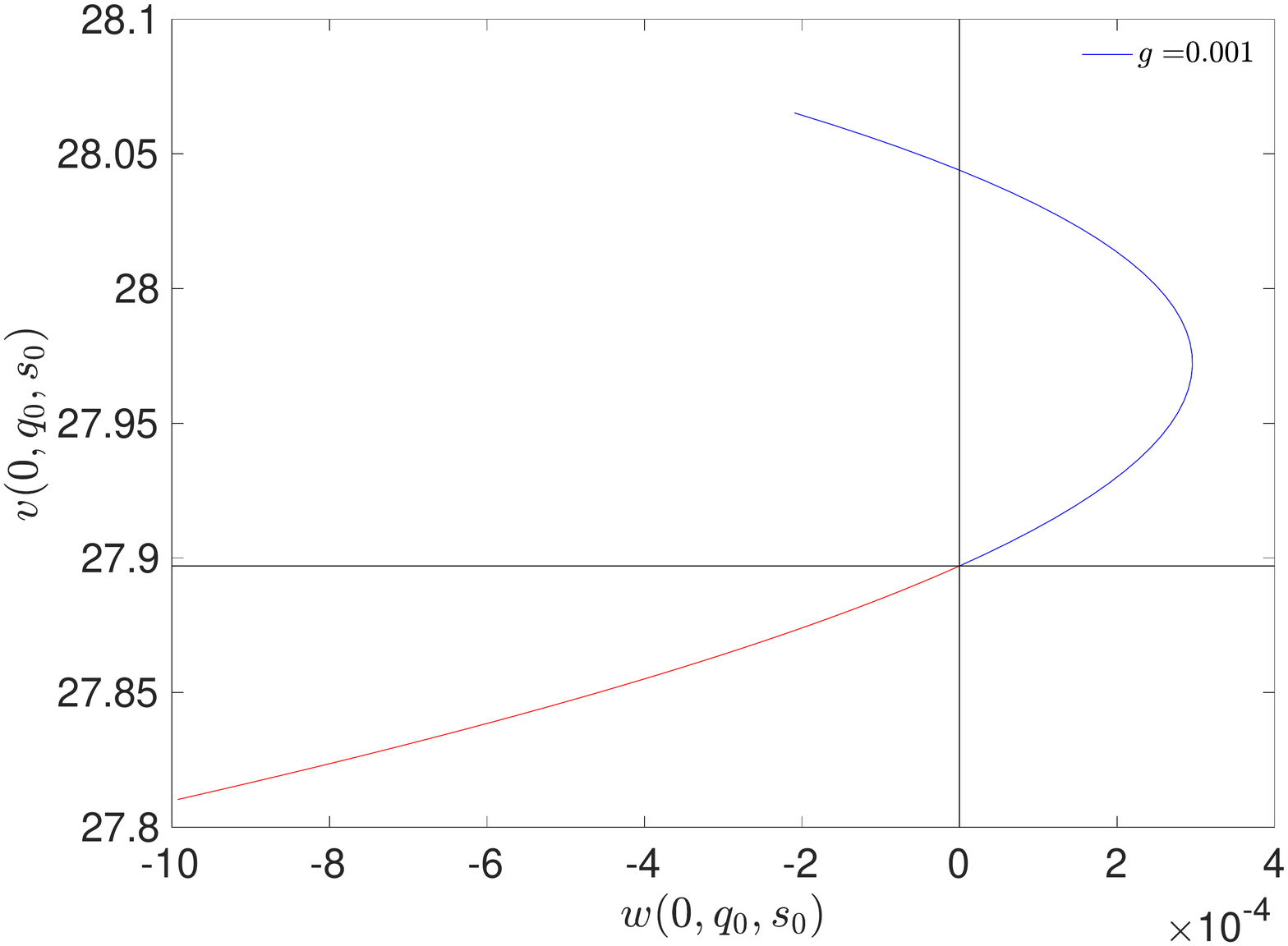}
\end{tabular}
\caption{{\small Values of $v(0,q_0,s_0)$ and $w(0,q_0,s_0)$ as a function of the trading position $\lambda$, $\lambda >0$ in red, $\lambda<0$ in blue, for $q_0 =  s_0/2 a$, and for different values of $a$.}}
\label{fig:ana}
\end{center}
\end{figure}

We have seen that the producer can drive the price at maturity at a level that would make her derivative position a profitable trade for her, providing her with an efficient tool in this asymmetric game of price manipulation. In this situation, considered the potential strong asymmetry of power in this game, a natural question on whether an exchange level $\lambda$ that would make both players better off exchanging might arise. Figure~\ref{fig:v0w0} (right) gives the value functions of the producer and of the trader at time zero as a function of the derivative position. We observe that even if the value function of the producer exhibits the same pattern as in the first two models, her expected profit is now considerably reduced due to the counteraction of the trader. Besides, the value function of the trader is concave and admits an optimum at a position that makes the producer worse off trading. Further, we observe that in this situation, with a zero initial rate of production, neither the producer nor the trader are better off trading.

But, if we consider that the production rate starts at its optimal stationary level $q^\star$, we find that whatever the market power of the producer, there is an exchange position making both the producer and the trader better off than not making a trade.  Figure~\ref{fig:ana} presents the value functions of the producer and the trader at initial time for different values of the market power parameter $a$ and different cost of intervention for the trader $g$. In each case, we chose as an initial production rate $q_0$ $=$ $q^\star$ $=$ $s_0/(2a)$, the stationary level of production, avoiding in this way the transitory phase to optimal production rate. When the producer and the trader do not trade ($\lambda=0$), their respective value $v(0,q_0,s_0)$ and $w(0,q_0,s_0)$ stand at the intersection of the black axis. As the trader starts to sell the derivative, $\lambda$ becomes negative and both values are greater than when $\lambda=0$, showing that both are better off making the trade.  Further, we observe that both are worse off in the case where the trader buys the derivative from the producer.

\clearpage
\bibliographystyle{plain}

 \end{document}